\documentclass[journal]{IEEEtran}

\usepackage{amsmath,amsfonts,amssymb}
\usepackage{amsthm}
\usepackage{array}
\usepackage{textcomp}
\usepackage{stfloats}
\usepackage{url}
\usepackage{graphicx}
\usepackage{cite}
\usepackage{bm}
\usepackage{braket}
\usepackage{mathtools}
\usepackage{booktabs}
\usepackage{rotating}
\usepackage{afterpage}
\usepackage{subcaption}
\usepackage{tikz}

\graphicspath{{figs/}}

\newtheorem{definition}{Definition}
\newtheorem{lemma}{Lemma}
\newtheorem{theorem}{Theorem}
\newtheorem{corollary}{Corollary}

\newcommand{\F}{\mathbb{F}}
\newcommand{\CNOT}{\mathrm{CNOT}}
\newcommand{\diag}{\operatorname{diag}}
\newcommand{\GP}{\mathrm{GP}}
\newcommand{\GPS}{\mathrm{GPS}}
\newcommand{\GPF}{\mathrm{GPF}}
\newcommand{\Dist}{\mathrm{Dist}}

\begin{document}

\title{Asymptotically Optimal Circuit Depth for Diagonal Unitary Synthesis and Compilation on Two-Dimensional Grids}

\author{
	Chengzhuo~Xu,
	Xiao~Chen,
	Zhihao~Liu
    and~Zhigang~Li
	\thanks{Draft Manuscript, June~2026;}
	\thanks{Chengzhuo Xu is with School of Computer Science and Engineering, Southeast University, Nanjing, 211189, China  (e-mail: 230228516@seu.edu.cn).}
	\thanks{Xiao Chen is with College of Information Engineering, China Jiliang University, Hangzhou, 310018, China (e-mail: 24a0305209@cjlu.edu.cn).}
	\thanks{Zhihao Liu is with School of Computer Science and Engineering, Southeast University, Nanjing, 211189, China and Key Laboratory of Computer Network and Information Integration (Southeast University), Ministry of Education, Nanjing, 211189, Jiangsu, China (email: liuzhtopic@163.com)}
	\thanks{Zhigang Li is with Nanjing Meteorology Bereau, Nanjing, 210019, Jiangsu, China (email: zhigangyangquan@163.com)}
    \thanks{(Corresponding author: Zhihao Liu.)}%
}

\maketitle
\raggedbottom

\begin{abstract}
Diagonal unitaries are a fundamental but resource-intensive class of quantum operations, arising as the phase separators of QAOA and the time-evolution blocks of Hamiltonian simulation. Under all-to-all connectivity their optimal depth is established, but on nearest-neighbor hardware general-purpose compilers fall back on heuristic search, which yields no analyzable cost bound and becomes intractable at the very sizes where depth is the bottleneck. We address synthesis and compilation jointly. On the synthesis side, we develop a Gray-Path Framework (GPF) that realizes any $n$-qubit diagonal unitary in asymptotically optimal $R_z$ and CNOT depth $O(2^n/n)$ without ancillas. Our main result is that compiling GPF onto a two-dimensional nearest-neighbor grid preserves this optimality: routing adds depth $\Theta(2^n/n)$ and gate count $\Theta(2^n)$. Because GPF fixes its entire interaction structure in advance, routing reduces to scheduling a known sequence, with no heuristic search. We give the construction both with and without ancillas: the ancilla-free, cost-optimized layout is a two-row grid, and a $2k$-row layout introduces a space--time tradeoff that cuts depth by $1/k$ while remaining asymptotically optimal for the enlarged register; both are deterministic and analyzed in closed form. The same complexity is also attained on a linear nearest-neighbor chain, so the preservation is topology-independent, holding on any architecture that contains such a chain. All routing bounds are closed-form, giving the concrete resource estimates that heuristic compilers cannot provide at scale.
\end{abstract}


\section{Introduction}
\label{sec:introduction}

\IEEEPARstart{D}{iagonal} unitaries are among the simplest algebraic classes of quantum operations, yet their synthesis is a recurring bottleneck in quantum algorithms. In QAOA the phase-separator operators are diagonal in the computational basis; in Hamiltonian simulation commuting Pauli-$Z$ terms produce diagonal time-evolution blocks after basis changes; and in instantaneous quantum polynomial-time circuits the computational content is carried by commuting phase operations \cite{Farhi2014QAOA,Bremner2011IQP}. These examples make diagonal-unitary synthesis a natural target for compiler-level depth optimization. Existing Walsh- and phase-polynomial-based methods synthesize arbitrary $n$-qubit diagonal unitaries without ancillas \cite{Welch2014Diagonal}, and $O(2^n/n)$ is the asymptotically optimal depth for this task \cite{Sun2023Depth,Zhang2024DiagonalDepth}.

This paper studies the circuit realization of diagonal unitaries in two stages, synthesis and compilation, and tracks the leading depth constants in closed form across both; treating the two together is the integrated view of circuit optimization emphasized in a recent survey \cite{Yan2024Overview}. The first stage is \emph{synthesis} under all-to-all logical connectivity: we ask whether the $O(2^n/n)$ scale can be reached by an ancilla-free family whose CNOT skeleton is transparent and whose leading constants can be computed. We answer this with a Gray-Path Framework (GPF) that converts diagonal synthesis into a traversal of Walsh parity modes: a CNOT network changes the parity checked by a single-qubit phase gate, so a qubit track that visits all nonzero Walsh modes synthesizes the target unitary up to global phase, and recursive multi-track mixing exposes the parallelism needed for $O(2^n/n)$ depth.

The second stage is \emph{compilation} onto a two-dimensional nearest-neighbor grid. Real hardware rarely satisfies all-to-all connectivity: once only nearest-neighbor two-qubit gates are permitted, the logical CNOT layers must be mapped, routed, and scheduled. General-purpose compilers treat routing as a generic problem solved by heuristic search such as SABRE \cite{Li2019SABRE}. For GPF this is unnecessary, because its interaction structure is fully determined by the synthesis construction: the entire sequence of two-qubit interactions is known before any qubit is placed, so routing is not a matter of discovering which gates to schedule but of scheduling a known sequence optimally, and we show this admits a closed-form optimum. We prove that compiling GPF onto a nearest-neighbor two-row grid preserves its asymptotic complexity: routing adds depth $\Theta(2^n/n)$ and gate count $\Theta(2^n)$, with closed-form leading constants. Restricted topology costs a definite constant factor, not a change of complexity class. The recursion already runs its deeper blocks on single rows; folding the top level as well places the entire construction on a one-dimensional line, which extends the preservation to be topology-independent at twice the two-row constant, while the two-row grid remains the primary, cost-optimized geometry and the substrate for an ancilla depth--space extension.

The main contributions fall in two halves: the synthesis half supplies the enabling ancilla-free construction with explicit constants, while the compilation half, where the paper's main result lies, shows that this optimality survives restricted-topology routing with explicit leading constants.

\noindent\textbf{Synthesis (Section~\ref{sec:gpf-framework}).}
\begin{itemize}
    \item We formulate diagonal-unitary synthesis as Walsh parity-phase synthesis and establish the exact rule by which a CNOT skeleton conjugates a parity phase (Lemma~\ref{lem:mode-migration}).
    \item We define a basic Gray Path construction and a recursive Gray-Path Framework with local mixing modules $\GPS$ and $\GPS^*$ that realize all nonzero Walsh modes without ancillas and reach $R_z$ and CNOT depth $O(2^n/n)$.
    \item We solve the depth recurrences in closed form: a single product $C\approx 3.40147$ governs the leading constants, with an explicit upper envelope $2(C-1)\approx 4.80294$ holding for all $n$ (Theorem~\ref{thm:balanced-constants}).
\end{itemize}

\noindent\textbf{Compilation (Sections~\ref{sec:placement}--\ref{sec:full}).}
\begin{itemize}
    \item \textbf{Restricted-topology routing preserves the complexity, with explicit leading constants.} On a nearest-neighbor two-row grid, routing GPF adds depth $\Theta(2^n/n)$ and gate count $\Theta(2^n)$ (Theorems~\ref{thm:route-depth}--\ref{thm:route-count}), so restricted topology costs only a definite constant factor. Because the interaction structure is known in advance, the optimal schedule is closed-form rather than search-based, with the per-block cost fixed by a single constant $\mu_{\mathrm{jump}}$ (Theorem~\ref{thm:gps2l}); its leading constants are the explicit numbers a designer needs for resource estimation and topology selection.
    \item \textbf{Closed-form overhead, no search.} On its native two-row layout the construction's normalized routing overhead is \emph{bounded} and available in closed form, whereas general-purpose heuristic routers rely on search that becomes infeasible at these sizes (Section~\ref{sec:comparison}).
    \item \textbf{Topology-independence.} Folding the top level onto a single row gives a one-dimensional-line realization, so the preservation extends to any topology containing a linear chain (Theorem~\ref{thm:gps1l}, Corollary~\ref{cor:topology}), at twice the two-row constant.
    \item \textbf{Space--time tradeoff.} A fan-out replication of the phase layers across additional rows reduces depth by $1/k$ with only polynomial ancilla overhead (Section~\ref{sec:fanout}).
\end{itemize}

We next situate these contributions in the adjacent literature.

\textbf{Diagonal-unitary synthesis.} The general problem of synthesizing an arbitrary unitary into elementary gates was placed on a systematic footing by Shende, Bullock, and Markov \cite{Shende2006Synthesis}; diagonal unitaries form a structured subclass, for which classical methods exploit the Walsh or phase-polynomial representation of the phase function. Bullock and Markov established asymptotically efficient decompositions for arbitrary diagonal computations \cite{Bullock2004Diagonal}; later constructions refined ancilla-free diagonal-unitary synthesis and connected it to phase polynomials, CNOT networks, and Gray-code traversals \cite{Welch2014Diagonal,Amy2018CNOTPhase,Vandaele2022PhasePolynomial,Zhang2024DiagonalDepth}, and depth-optimal synthesis results identify $O(2^n/n)$ as the natural asymptotic scale for such tasks \cite{Sun2023Depth}. The synthesis half of this paper contributes a member of this family with explicit leading constants; the compilation half then asks what restricted-topology routing adds to it.

\textbf{Routing as matchings.} The view of restricted-topology compilation as ``routing via matchings''---rooted in the linear-depth fermionic swap network of Kivlichan et al.\ \cite{Kivlichan2018SwapNetwork}---was put in general form by O'Gorman et al.\ \cite{OGorman2019SwapNetworks}, who showed that any unordered set of $k$-qubit gates can be ordered and parallelized on a linear arrangement in $O(n^{k-1})$ depth, and pointed to QAOA phase separators and fermionic SWAP networks as instances. Our abstraction is the same in spirit, but specialized: GPF requests a \emph{structured} family of cyclic-offset matchings, for which we obtain not a generic upper bound but a closed-form, constant-factor-optimal deterministic schedule (instantiated with a low-jump Gray code), and we treat two geometries (two-row and folded single-row) uniformly.

\textbf{Structure-aware routing.} Exploiting known circuit structure rather than generic search is an active and effective line for QAOA and related circuits: edge-coloring-based routing \cite{Kotil2023QAOARouting}, pattern-based mapping to linear-nearest-neighbor layouts \cite{Zhu2024Coqa}, fermionic SWAP networks with circuit averaging \cite{Hashim2022FermionicSwap}, parity-twine/SWAP-network encodings \cite{Montanez2025ParityTwine}, and structured, non-search compilation guaranteeing linear depth for QAOA circuits on multi-dimensional architectures \cite{Jin2021StructuredQAOA}. Such linear-depth guarantees rest on the QAOA phase separator being a sparse, few-body diagonal; the present work instead targets the full dense $2^n$-mode diagonal-unitary spectrum, for which $\Theta(2^n/n)$ depth is intrinsic, carrying the structure-exploiting principle there and supplying the missing closed-form complexity.

\textbf{Steiner-tree CNOT routing.} The methodologically closest line routes \emph{arbitrary} CNOT / phase-polynomial circuits on a restricted coupling graph by Gaussian elimination constrained to Steiner trees: Steiner-Gauss \cite{Kissinger2020SteinerGauss}, the Steiner-tree optimizations of Nash et al.\ \cite{Nash2020NISQ}, the dynamically-remapping PermRowCol \cite{Griend2023PermRowCol} (itself deterministic, not search-based), and depth-oriented linear-reversible synthesis \cite{Brugiere2022Depth}. These methods are general and instance-specific---each target parity set is routed from scratch by elimination, with no closed-form cost. GPF's matchings, by contrast, are a single \emph{known, structured} family (cyclic-offset), which is exactly what lets us replace per-instance elimination by one closed-form, constant-factor-optimal schedule fixed in advance.

\textbf{Connectivity and circuit complexity.} On the theoretical side, Yuan, Allcock, and Zhang \cite{Yuan2024Connectivity} showed that limited connectivity (e.g., a 1-D chain) does not blow up the asymptotic complexity of several structured families, including diagonal unitaries, and that ancillas restore polynomial depth. That is an existence/asymptotic statement; we complement it with an explicit, deterministic construction whose leading constants are given in closed form, where an asymptotic order alone gives nothing concrete for planning at the sizes ($n>16$) of interest. Related results on CNOT-circuit cost under limited connectivity \cite{Wu2023CNOTLimited} and the space--depth trade-off of CNOT circuits \cite{Jiang2020SpaceDepth} share the CNOT-skeleton perspective we adopt.

\textbf{Deterministic structured compilation.} The thesis that hardware-structured, non-iterative compilation can outperform repeated global search has recent echoes on other platforms, e.g., the deterministic zoned-architecture compiler for neutral-atom arrays \cite{Huang2026ZAP}. On the ancilla side, adjustable-depth diagonal-operator circuits with space--time--accuracy trade-offs \cite{Zylberman2024Diagonal} are the closest relatives of our Section~\ref{sec:fanout} construction.

The remainder of the paper is organized as follows. Section~\ref{sec:preliminaries} sets up the shared foundations---the Walsh and CNOT-skeleton model, the depth model, and the transition sequence. Section~\ref{sec:gpf-framework} develops the synthesis half: the basic Gray Path, the recursive GPF, and the closed-form depth analysis. Sections~\ref{sec:placement}--\ref{sec:full} develop the compilation half: the two-row placement, the two routing-module theorems, and the complete assembly with its complexity and a comparison with general-purpose routers; Section~\ref{sec:fanout} develops the ancilla-assisted fan-out construction and Section~\ref{sec:discussion} concludes. Appendices collect the closed-form derivations and the full eight-qubit circuits.

\section{Preliminaries}
\label{sec:preliminaries}

\subsection{Diagonal Unitaries and Walsh Parity Phases}

Let $x,s\in\F_2^n$, and let $s\cdot x$ denote the binary inner product over $\F_2$. An arbitrary $n$-qubit diagonal unitary can be written as
\begin{equation}
    \Lambda_n(\bm{\theta})
    =
    \sum_{x\in\F_2^n} e^{i\theta_x}\ket{x}\bra{x}.
    \label{eq:diagonal}
\end{equation}
Directly synthesizing Eq.~\eqref{eq:diagonal} in the computational basis addresses individual basis states with multi-controlled phase gates, obscuring the operator's commutative structure.

We instead express the phase function in the Walsh basis. Define
\begin{equation}
    \widehat{\theta}_s
    =
    2^{-n}\sum_{x\in\F_2^n}(-1)^{s\cdot x}\theta_x .
    \label{eq:walsh-coef}
\end{equation}
Then
\begin{equation}
    \theta_x
    =
    \sum_{s\in\F_2^n}(-1)^{s\cdot x}\widehat{\theta}_s .
\end{equation}
Using $(-1)^b=1-2b$ for $b\in\{0,1\}$, the diagonal unitary decomposes into conditional parity phases. For $s\ne 0$, define
\begin{equation}
    P_s(\varphi_s)
    =
    \sum_{x\in\F_2^n}
    e^{i\varphi_s(s\cdot x)}
    \ket{x}\bra{x},
    \qquad
    \varphi_s=-2\widehat{\theta}_s .
    \label{eq:parity-phase}
\end{equation}
There exists a global phase $\gamma=\sum_s\widehat{\theta}_s$ such that
\begin{equation}
    \Lambda_n(\bm{\theta})
    =
    e^{i\gamma}\prod_{s\in\F_2^n\setminus\{0\}}P_s(\varphi_s).
    \label{eq:walsh-product}
\end{equation}
The global phase is not physically observable and will not be implemented by the circuit constructions below. The nonzero factors in Eq.~\eqref{eq:walsh-product} commute, so their synthesis order can be chosen to optimize routing cost.

Throughout this paper, $R_z(\phi)$ denotes $\diag(1,e^{i\phi})$, i.e.\ the phase gate $P(\phi)$. The symmetric rotation convention instead defines $\diag(e^{-i\phi/2},e^{i\phi/2})$; the two agree up to the global phase $e^{i\phi/2}$, which is unobservable and convenient to drop for parity-phase synthesis.

\subsection{CNOT Skeletons}

An ancilla-free CNOT circuit over $n$ qubits implements an invertible linear transformation $A\in\mathrm{GL}(n,2)$ on computational basis states \cite{Patel2008LinearReversible}:
\begin{equation}
    C_A\ket{x}=\ket{Ax}.
\end{equation}
We call $C_A$ the CNOT skeleton. The following conjugation rule is the algebraic engine of the construction.

\begin{lemma}[Conjugation of a parity phase by a CNOT skeleton]
\label{lem:mode-migration}
For every $A\in\mathrm{GL}(n,2)$, $s\in\F_2^n$, and phase angle $\varphi$,
\begin{equation}
    C_A^\dagger P_s(\varphi) C_A
    =
    P_{A^T s}(\varphi).
    \label{eq:cnot-conjugation}
\end{equation}
\end{lemma}

\begin{proof}
For any basis state $\ket{x}$, the phase contributed by the left-hand side depends on $s\cdot(Ax)$. Since $s\cdot(Ax)=(A^Ts)\cdot x$ over $\F_2$, the conjugated operator applies phase $e^{i\varphi}$ exactly on those $x$ satisfying $(A^Ts)\cdot x=1$, which is $P_{A^Ts}(\varphi)$.
\end{proof}

In particular, applying $R_z(\varphi)$ to qubit $k$ realizes $P_{e_k}(\varphi)$. Under a CNOT skeleton $C_A$, the same inserted phase becomes $P_{A^Te_k}(\varphi)$. Hence one selected qubit track can realize different Walsh modes as the corresponding skeleton mode evolves.

This separates every construction in this paper into two layers. A \emph{CNOT skeleton} is a sequence of skeleton updates that carries the phase-insertion track through a chosen list of Walsh modes; together with the \emph{insertion slots} it exposes---one per visited mode---the skeleton depends only on $n$ and the chosen traversal, not on the target unitary. Filling a slot with $R_z(\varphi_s)$ while the current mode equals $s$ realizes the factor $P_s(\varphi_s)$ of Eq.~\eqref{eq:walsh-product}; only the angles $\varphi_s=-2\widehat{\theta}_s$ are instance-specific. A skeleton therefore becomes a concrete diagonal-unitary circuit only after this \emph{phase insertion}, and a single skeleton synthesizes every $n$-qubit diagonal unitary. In the figures below we draw each insertion slot of a skeleton as a dashed split line; the basic $\GP$ is the special case in which every slot is filled by a phase.

\subsection{Depth Model and Hardware Assumption}

The reported costs are circuit \emph{layer depths}: the number of parallel gate layers on the critical path, with only two-qubit gates constrained by connectivity (single-qubit phase gates apply on any qubit in place). Each cost is a layer depth $D$ or a gate count $S$, carrying a superscript that names the construction (kept explicit throughout, never blank or primed) and a subscript that selects one of three independently counted tracks:
\begin{itemize}
    \item $D_z,\,S_z$: the depth and gate count of the phase ($R_z$) layers;
    \item $D_{\CNOT},\,S_{\CNOT}$: the depth and gate count of the CNOT layers;
    \item $D_{\mathrm{route}},\,S_{\mathrm{route}}$: the depth and gate count of the CNOTs added by routing on a restricted topology.
\end{itemize}

The two halves of the paper differ only in the connectivity they assume: synthesis (Section~\ref{sec:gpf-framework}) under all-to-all logical connectivity, and compilation (Sections~\ref{sec:placement}--\ref{sec:full}) restricting two-qubit gates to the nearest-neighbor edges of a two-dimensional grid.

\subsection{The Gray-Code Transition Sequence and Its Two Cost Functionals}
\label{sec:graycode}

Both halves of this paper rest on a single combinatorial object. A \emph{Gray-code transition sequence} on $m$ coordinates---abbreviated \emph{transition sequence} below---is the list of flipped coordinates of a single-bit-flip (Gray) traversal of $\F_2^m$: at step $t$ the traversal flips coordinate $f_t\in\{0,\ldots,m-1\}$. It appears in two forms. The basic Gray Path traverses the full space ($m=n$) along an open Hamiltonian path; the local mixing modules GPS/GPS$^*$ are driven by a \emph{closed} cycle $G=(f_0,f_1,\ldots,f_{2^m-1})$ on the high subspace ($m=n_H$), which must visit every mode and return to mode $0$ after all flips (Appendix~\ref{app:lowjump}). The synthesized unitary is independent of the choice of $G$; the \emph{cost}, however, is governed by two functionals of $G$ that we separate here because they pull in opposite directions.

\emph{Fold factor $\kappa(G)$.} Partition the $2^m$ flips of $G$ into consecutive groups of four; write the flips of group $g$ as $(g_1,g_2,g_3,g_4)$ and assign it the cost
\begin{equation*}
    c_g=
    \begin{cases}
        3, & g_1=g_3 \ \text{ or }\ g_2=g_4,\\
        4, & \text{otherwise}.
    \end{cases}
\end{equation*}
A group is thus cheap ($c_g=3$) exactly when its first and third flips coincide or its second and fourth do (call such a group \emph{matched}), and costs $c_g=4$ when neither pair repeats. The fold factor is the average of $c_g$ over the $2^{m-2}$ groups,
\begin{equation}
    \kappa(G)=\frac{\sum_{g=1}^{2^{m-2}} c_g}{2^{m-2}}\;\in[3,4],
    \label{eq:kappa-def}
\end{equation}
so the range $[3,4]$ is immediate from $c_g\in\{3,4\}$, with $\kappa=3$ when every group is matched and $4$ when none is; the reading of $c_g$ as a group's CNOT-to-phase-layer ratio is established in Appendix~\ref{app:kappa-bound}. The binary reflected Gray code (BRGC) flips the lowest coordinate at every odd step, so every group is matched and $\kappa\equiv 3$ exactly; an irregular sequence mixes matched and unmatched groups, giving $\kappa\in(3,4)$ with no closed form.

\emph{Jump density $\mu_{\mathrm{jump}}(G)$.} Place the $m$ coordinates on a ring and let $\Dist_m(a,b)=\min(|a-b|,\,m-|a-b|)$ be the cyclic distance between two ring positions. The jump density is the average cyclic distance between consecutive flips,
\begin{equation}
    \mu_{\mathrm{jump}}(G)=\frac{1}{2^m}\sum_{t=0}^{2^m-1}\Dist_m\!\left(f_t,\,f_{(t+1)\bmod 2^m}\right).
    \label{eq:mu-def}
\end{equation}
Adjacent flips differ, so every term is at least $1$ and $\mu_{\mathrm{jump}}\ge 1$, with equality reachable only at small $m$ (Appendix~\ref{app:lowjump}). The BRGC scores poorly here, whereas a sequence whose successive flips stay close on the ring keeps $\mu_{\mathrm{jump}}$ near this lower bound.

For the depth of the final compiled circuit both functionals matter, and one would like each as small as possible: a smaller $\kappa$ lowers the synthesis-side CNOT depth, and a smaller $\mu_{\mathrm{jump}}$ lowers the compilation-side routing depth. The two pull in opposite directions, however, and are hard to keep small simultaneously. Since routing contributes the dominant constant to the compiled depth, this paper adopts a \emph{low-jump} sequence (one that keeps $\mu_{\mathrm{jump}}$ small) as its default throughout, with explicit construction in Appendix~\ref{app:lowjump}. Whether a valid sequence with strictly smaller $\mu_{\mathrm{jump}}$ exists is an open problem, and any such improvement would directly lower the final compiled depth.

\section{The Gray-Path Framework}
\label{sec:gpf-framework}

\subsection{Basic Gray Path Synthesis}
\label{sec:gp}

We order the Walsh modes by a transition sequence (Section~\ref{sec:graycode}), in its open-path form on the full space $\F_2^n$. We adopt the BRGC as a representative valid transition sequence. Its single-bit-flip property is precisely what is needed for low-cost mode traversal: moving from one Walsh mode to the next changes only one coordinate and can be implemented by a single CNOT update of the current skeleton. We use the BRGC as the running example, but the construction depends only on the single-bit-flip property: any such transition sequence yields the \emph{same} diagonal unitary. The choice of sequence affects cost alone, through the two functionals of Section~\ref{sec:graycode} (the fold factor $\kappa$, Eq.~\eqref{eq:kappa-def}, and the jump density $\mu_{\mathrm{jump}}$, Eq.~\eqref{eq:mu-def}): of the circuit depths only the $\GPF^*$ CNOT depth varies, through $\kappa$. Gray-code based traversals have also appeared in general multiqubit-gate decompositions \cite{Mottonen2004Multiqubit}.

\begin{definition}[Gray Path]
Fix a qubit track $a$ as the phase insertion track. A Gray Path $\GP(n)$ is the circuit obtained by traversing the nonzero Walsh modes in the order of a transition sequence (e.g.\ a BRGC). At each step, if the next mode differs from the current mode in bit $f$, one applies a CNOT update controlled on qubit $f$ and targeted on track $a$, and then inserts the phase gate $R_z(\varphi_s)$ on track $a$.
\end{definition}

\begin{theorem}[Correctness of $\GP(n)$ up to global phase]
\label{thm:gp-correctness}
For every diagonal unitary $\Lambda_n(\bm{\theta})$, the Gray Path construction realizes Eq.~\eqref{eq:walsh-product} after removing the global phase $e^{i\gamma}$, and uses no auxiliary qubits.
\end{theorem}

\begin{proof}
By Lemma~\ref{lem:mode-migration}, the phase gate inserted on track $a$ realizes the Walsh mode equal to the current skeleton mode associated with that track. The transition sequence visits every nonzero mode exactly once, and each transition updates the skeleton to the next mode. Since the operators $P_s(\varphi_s)$ commute, their order does not affect the resulting diagonal unitary. The zero mode contributes only to $e^{i\gamma}$ in Eq.~\eqref{eq:walsh-product}, so it can be omitted.
\end{proof}

The basic $\GP(n)$ construction is conceptually useful but not depth-optimal: it has a single active track and therefore uses $2^n$ phase layers and $2^n$ CNOT layers. Figures \ref{fig:gp-brgc} and \ref{fig:gp-alt} show representative $n=3$ Gray Path circuits.

\begin{figure}[!t]
    \centering
    \begin{subfigure}[b]{\columnwidth}
        \centering
        \includegraphics[width=\columnwidth]{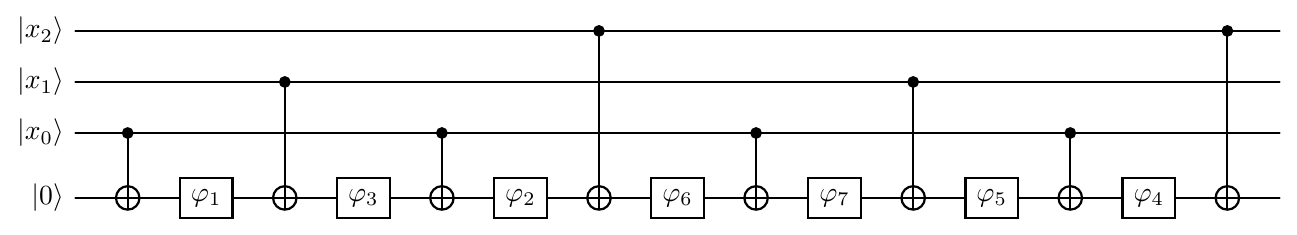}
        \caption{BRGC traversal.}
        \label{fig:gp-brgc}
    \end{subfigure}
    \\[0.6em]
    \begin{subfigure}[b]{\columnwidth}
        \centering
        \includegraphics[width=\columnwidth]{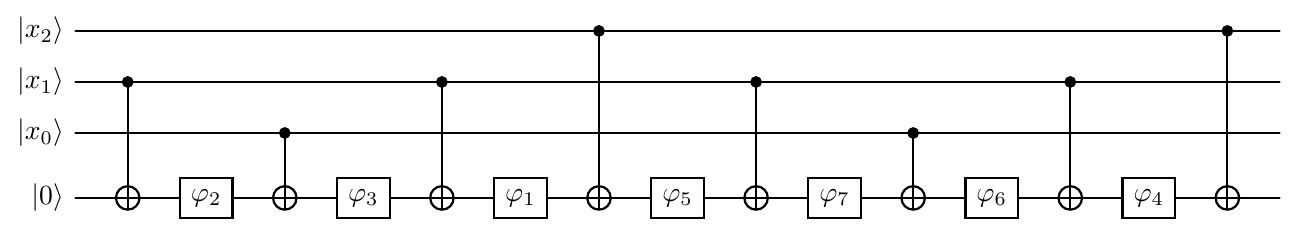}
        \caption{Transition sequence $[1,0,1,2,1,0,1,2]$.}
        \label{fig:gp-alt}
    \end{subfigure}
    \caption{Two equivalent basic Gray Path circuits for $n=3$: different transition sequences realize the same commuting product of Walsh phases. Phase labels abbreviate the angles $\varphi_s$.}
\end{figure}

\subsection{Recursive Multi-Track Cover}
\label{sec:gpf}

The main limitation of $\GP(n)$ is that it exposes only one active phase track. The Gray-Path Framework (GPF) replaces this single-track traversal by a recursive multi-track cover of the Walsh space. The guiding principle is to divide the qubits into two nearly equal subspaces and use one subspace to generate simultaneous tracks through the other.

\subsubsection{Recursive Decomposition}

Partition the qubit set into a high subspace $\mathcal{Q}_H$ and a low subspace $\mathcal{Q}_L$, with
\begin{equation}
    n_H=\lceil n/2\rceil,\qquad n_L=\lfloor n/2\rfloor .
\end{equation}
Every Walsh mode $s\in\F_2^n$ decomposes uniquely as
\begin{equation}
    s=s_H\oplus s_L,\qquad
    s_H\in\F_2^{n_H},\quad s_L\in\F_2^{n_L}.
\end{equation}
We name the high-subspace qubits $\mathrm{H}_0,\ldots,\mathrm{H}_{n_H-1}$ and the low-subspace qubits $\mathrm{L}_0,\ldots,\mathrm{L}_{n_L-1}$; this high/low pair is the primary qubit label used throughout the paper. The circuits of this section are drawn with it, and on the two-row grid the same $\mathrm{H}_i/\mathrm{L}_i$ name the physical sites the qubits occupy, so no separate physical-site notation is introduced there. The remaining labels appearing later---the global logical index $x_i$, the offset-ring occupancy $h_i/l_i$, and the recursive coordinate $Q_i$---are sub-algorithm-specific relabelings, collected in Table~\ref{tab:notation}.

The GPF recursion separates two cases. When $s_L=0$, the remaining modes live entirely in $\mathcal{Q}_H$ and can be handled recursively. When $s_L\ne 0$, the low subspace selects a family of tracks and the high subspace supplies the local Gray traversal. This is the source of the $n$-fold parallelism in the construction.

\subsubsection{Local Mixing by $\GPS$ and $\GPS^*$}

The local mixing module $\GPS$ couples the high and low subspaces so that several Gray paths advance simultaneously (figure \ref{fig:gps}). Let $f_t$ be the high coordinate flipped at step $t$ of the high-space Gray traversal (transition sequence). Rather than recording a separate control for each track, $\GPS$ realizes this step as a single cyclic-offset matching between $\mathcal{Q}_H$ and $\mathcal{Q}_L$: for each low-space track $j$ it inserts one CNOT whose control--target pair is
\[
    (\mathrm{ctrl},\,\mathrm{tgt})=\bigl((f_t+j)\bmod n_H,\ j\bigr),\qquad j=0,1,\dots,n_L-1,
\]
with the control in $\mathcal{Q}_H$ and the target in $\mathcal{Q}_L$. Because $n_L\le n_H$, the shifted controls $(f_t+j)\bmod n_H$ are distinct across $j$, so the $n_L$ CNOTs act on disjoint qubits and form a single layer. Each step of the high-space traversal thus becomes one parallel matching layer and exposes one insertion slot shared across the $n_L$ low tracks; over its $2^{n_H}$ steps the $\GPS$ skeleton has $2^{n_H}$ CNOT layers and $2^{n_H}$ insertion slots. Filling those slots contributes $2^{n_H}$ phase layers at no extra depth, so $\GPS$ carries the same layer counts as $\GP$.

The variant $\GPS^*$ (figure \ref{fig:gps-star}) alternates forward and reverse local Gray paths. It reduces the number of phase layers at the price of additional CNOT fold-back layers. Informally, $\GPS$ is balanced for phase and CNOT depth, while $\GPS^*$ is phase-prioritized.

\begin{figure}[!t]
    \centering
    \begin{subfigure}[b]{\columnwidth}
        \centering
        \includegraphics[width=\columnwidth]{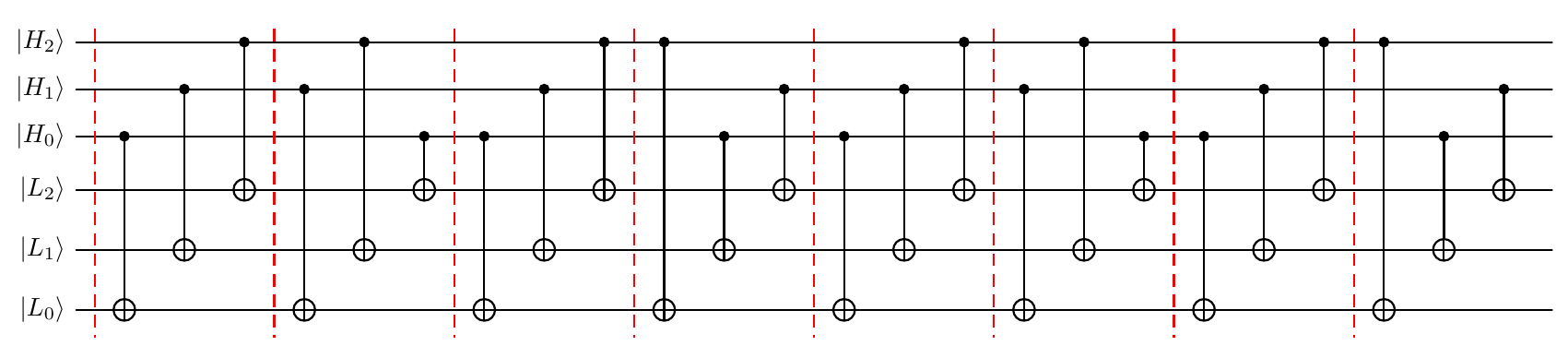}
        \caption{$\GPS$: each step is a cyclic-offset matching, so its CNOTs form one parallel layer.}
        \label{fig:gps}
    \end{subfigure}
    \\[0.6em]
    \begin{subfigure}[b]{\columnwidth}
        \centering
        \includegraphics[width=\columnwidth]{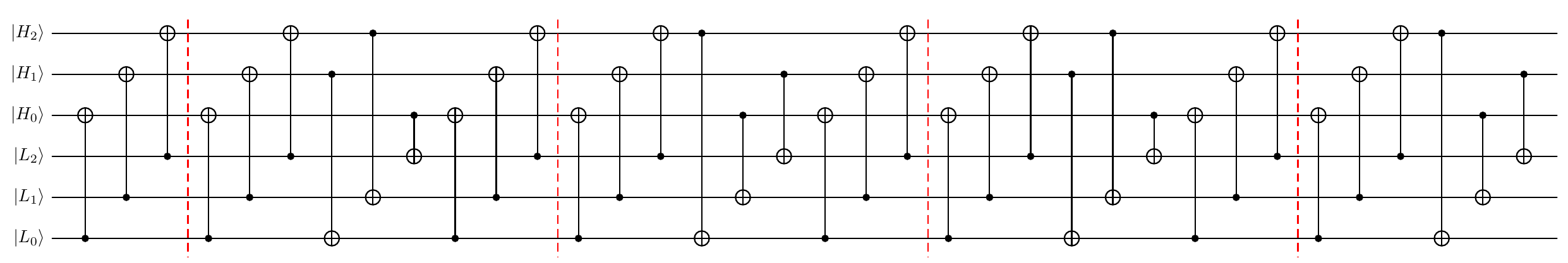}
        \caption{$\GPS^*$: forward and reverse Gray-path segments (phase-prioritized).}
        \label{fig:gps-star}
    \end{subfigure}
    \caption{CNOT skeletons of the two local mixing modules for $n_H=n_L=3$. Dashed split lines mark the insertion slots (one per step), empty until phases are inserted.}
\end{figure}

\subsubsection{Recursive Composition and Phase-Insertion Instances}

Figures \ref{fig:gpf} and \ref{fig:gpf-star} show the recursive composition of the $\GPF$ and $\GPF^*$ \emph{skeletons}. The split lines introduced by $\GPS$ are the recursion sites: rather than receiving a bare phase, each is where a low-subspace block attaches. Concretely, $\GPF_n$ alternates $\GPS_n$ mixing skeletons with low-subspace sub-blocks $\GPF^*_{n_L}$; the low subspace supplies the outer mode sequence while each $\GPS_n$ unfolds a full high-space Gray traversal around the current low mode. The pure high-space modes ($s_L=0$) are collected by a final $\GPF_{n_H}$ on the high subspace, and the recursion bottoms out at $n=1$, where a slot holds a single phase gate. The result is still a CNOT skeleton with insertion slots; it does not by itself implement a diagonal unitary.

A concrete diagonal unitary appears only after the slots are filled. Figure \ref{fig:gpf4-full} gives the complete $4$-qubit phase-insertion circuit obtained this way: every slot of the $n=4$ skeleton receives its $R_z(\varphi_s)$, the phase labels use four-bit Walsh-mode indices, and all $2^4-1$ nonzero Walsh modes appear exactly once. Because the skeleton is fixed by $n$ and the chosen traversal, the same circuit synthesizes any $4$-qubit diagonal unitary by changing only the angles. The full $8$-qubit instance is too wide for the main two-column text and is therefore placed in Appendix~\ref{app:gpf8}.

\begin{figure}[!t]
    \centering
    \begin{subfigure}[b]{\columnwidth}
        \centering
        \includegraphics[width=\columnwidth]{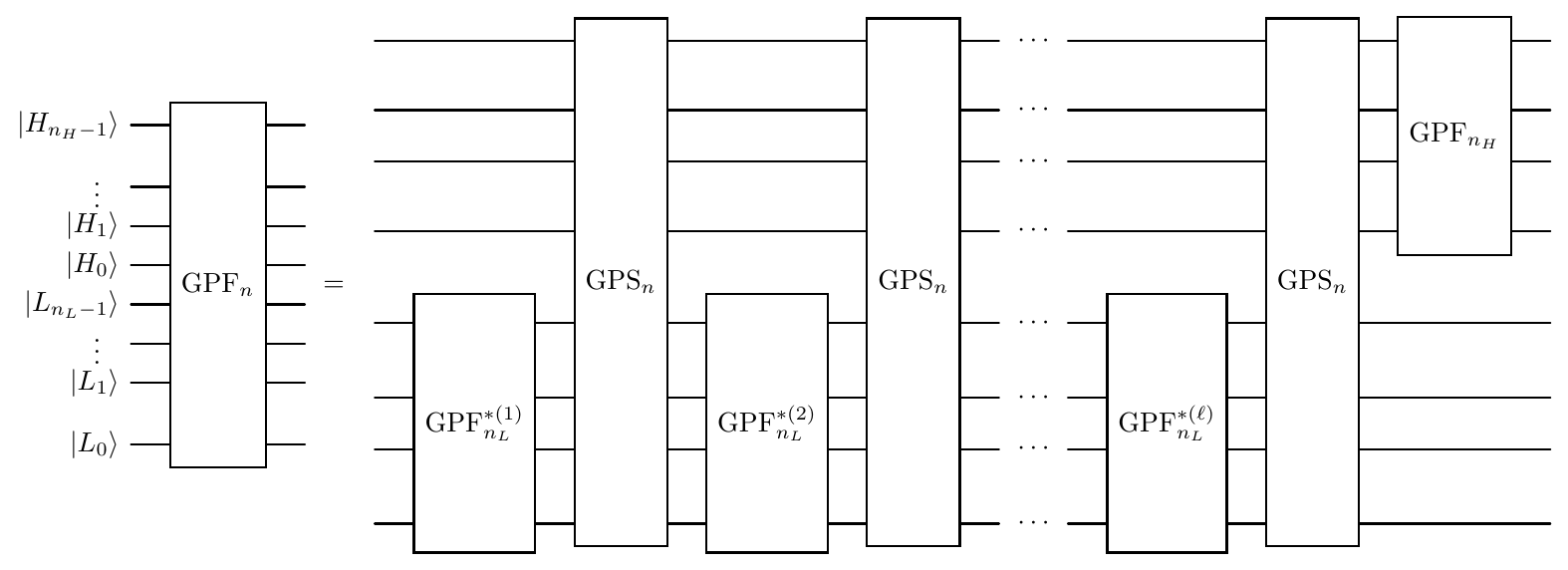}
        \caption{$\GPF$: the final high-space recursion handles the $s_L=0$ modes, the $\GPS_n$ mixings the $s_L\ne 0$ modes.}
        \label{fig:gpf}
    \end{subfigure}
    \\[0.6em]
    \begin{subfigure}[b]{\columnwidth}
        \centering
        \includegraphics[width=\columnwidth]{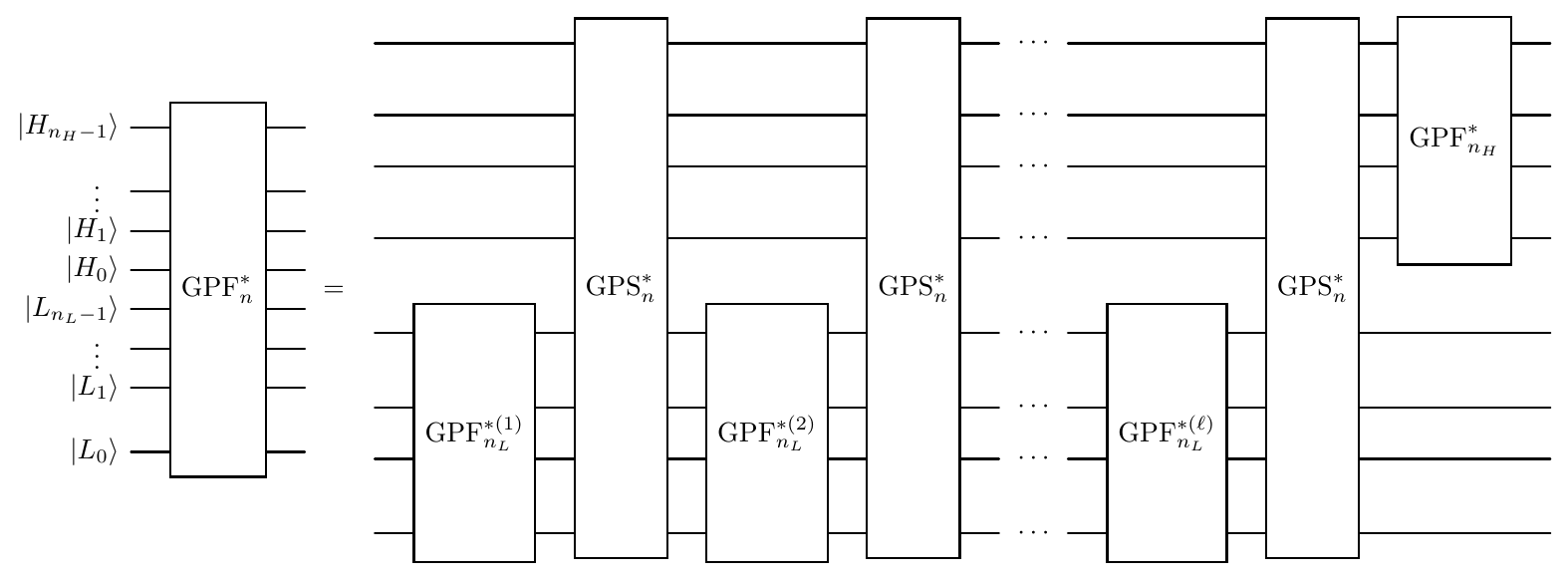}
        \caption{$\GPF^*$: uses the phase-prioritized $\GPS^*$ throughout the mixed-mode portion.}
        \label{fig:gpf-star}
    \end{subfigure}
    \caption{Recursive composition of the $\GPF$ and $\GPF^*$ skeletons; both carry no phase gates, which are inserted into the slots of their $\GPS$/$\GPS^*$ modules (Fig.~\ref{fig:gps}).}
\end{figure}

\begin{figure*}[!t]
    \centering
    \includegraphics[width=0.88\textwidth]{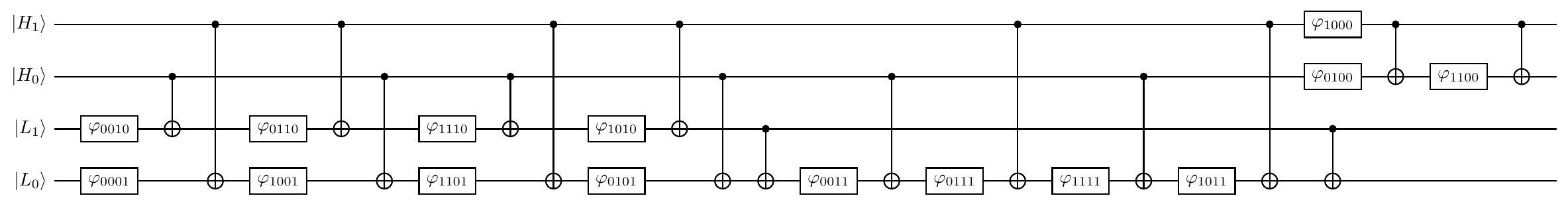}
    \caption{Complete GPF phase-insertion circuit for $n=4$, obtained by filling every insertion slot of the 4-qubit GPF skeleton. Each phase label $\varphi_s$ is indexed by a four-bit Walsh mode.}
    \label{fig:gpf4-full}
\end{figure*}

\subsection{Depth Analysis}
\label{sec:depth}

Let $D^{\GPF^*}_z(n)$ and $D^{\GPF^*}_{\CNOT}(n)$ denote the phase and CNOT depths of $\GPF^*$, respectively. Let $D^{\GPF}_z(n)$ and $D^{\GPF}_{\CNOT}(n)$ denote the corresponding depths of $\GPF$. The base case is immediate: when $n=1$, the only nonzero Walsh mode is $e_1$, so one phase layer and no CNOT layer are required:
\begin{equation}
\begin{aligned}
    D^{\GPF^*}_z(1)&=D^{\GPF}_z(1)=1,\\
    D^{\GPF^*}_{\CNOT}(1)&=D^{\GPF}_{\CNOT}(1)=0.
\end{aligned}
\label{eq:base}
\end{equation}

For $n\ge 2$, the recursive construction gives the following depth recurrences:
\begin{equation}\label{eq:recurrences}
\resizebox{0.88\columnwidth}{!}{$\displaystyle
\begin{aligned}
    &D^{\GPF^*}_z(n) = 2^{n_H-1}D^{\GPF^*}_z(n_L)+D^{\GPF^*}_z(n_H),\\
    &D^{\GPF^*}_{\CNOT}(n) = \kappa\,D^{\GPF^*}_z(n)+D^{\GPF^*}_{\CNOT}(n_H)+D^{\GPF^*}_{\CNOT}(n_L),\\
    &D^{\GPF}_z(n) = 2^{n_H}D^{\GPF^*}_z(n_L)+D^{\GPF}_z(n_H),\\
    &D^{\GPF}_{\CNOT}(n) = D^{\GPF}_z(n)+D^{\GPF}_{\CNOT}(n_H)+D^{\GPF^*}_{\CNOT}(n_L).
\end{aligned}$}
\end{equation}
The recurrences in~\eqref{eq:recurrences} formalize the construction: the low subspace determines how many mixed groups are needed, the high subspace supplies the local parallel traversal, and the pure high-space modes recurse. The coefficient $\kappa$ in the $D^{\GPF^*}_{\CNOT}$ line is the fold factor of Eq.~\eqref{eq:kappa-def}: $\kappa=3$ exactly for the BRGC, while a low-jump sequence gives $\kappa\in(3,4)$ with no closed form. The remaining three lines are Gray-code-independent.

To examine the leading constants we track the normalized depths $D^{\GPF}_z(n)/(2^n/n)$ and $D^{\GPF}_{\CNOT}(n)/(2^n/n)$, together with the analogous $\GPF^*$ ratios $D^{\GPF^*}_z(n)/(2^n/n)$ and $D^{\GPF^*}_{\CNOT}(n)/(2^n/n)$. Figures \ref{fig:gpf-depth-scaling} and \ref{fig:gpfstar-depth-scaling} plot these quantities obtained by direct recurrence evaluation. The curves are not monotone; they oscillate within dyadic scale intervals due to the rounding effects in $n_H=\lceil n/2\rceil$ and $n_L=\lfloor n/2\rfloor$. The lower envelope is attained on the balanced subsequence $n=2^\lambda$, where every recursive split is exact, and the upper envelope on the maximally unbalanced subsequence $n=2^\lambda-1$. As shown next, both envelopes converge to explicit constants, so the normalized depth is bracketed by two constants for all sufficiently large $n$.

\begin{figure}[!t]
    \centering
    \includegraphics[width=\columnwidth]{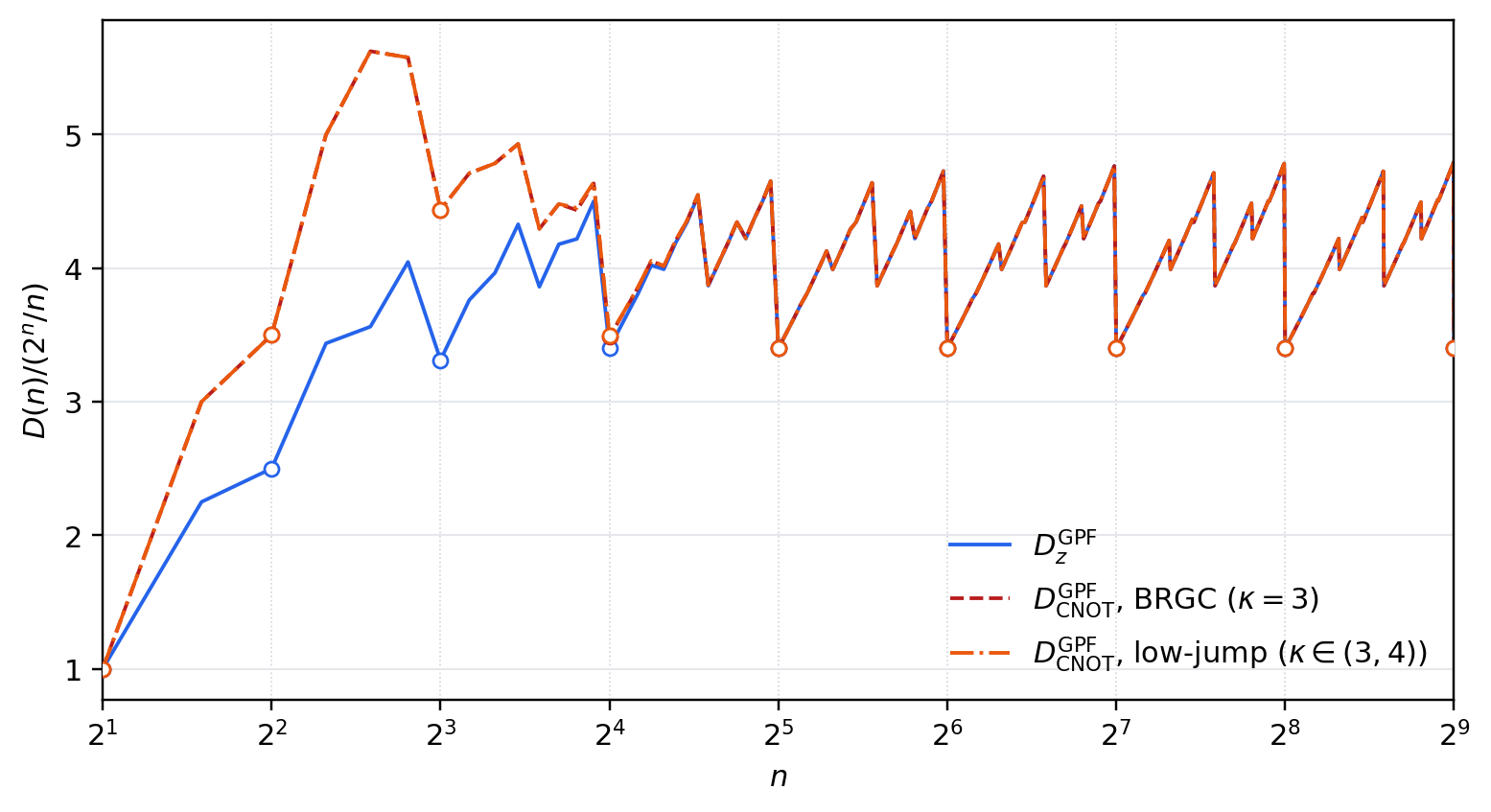}
    \caption{Normalized depth constants for $\GPF$. The phase-depth ratio is Gray-code-independent; the CNOT-depth ratio depends on the transition sequence only weakly, through the fold factor $\kappa$ (Section~\ref{sec:graycode}), so the BRGC and low-jump curves differ only at small $n$. Both converge to $C\approx 3.40147$ (hollow markers: $n=2^\lambda$).}
    \label{fig:gpf-depth-scaling}
\end{figure}

\begin{figure}[!t]
    \centering
    \includegraphics[width=\columnwidth]{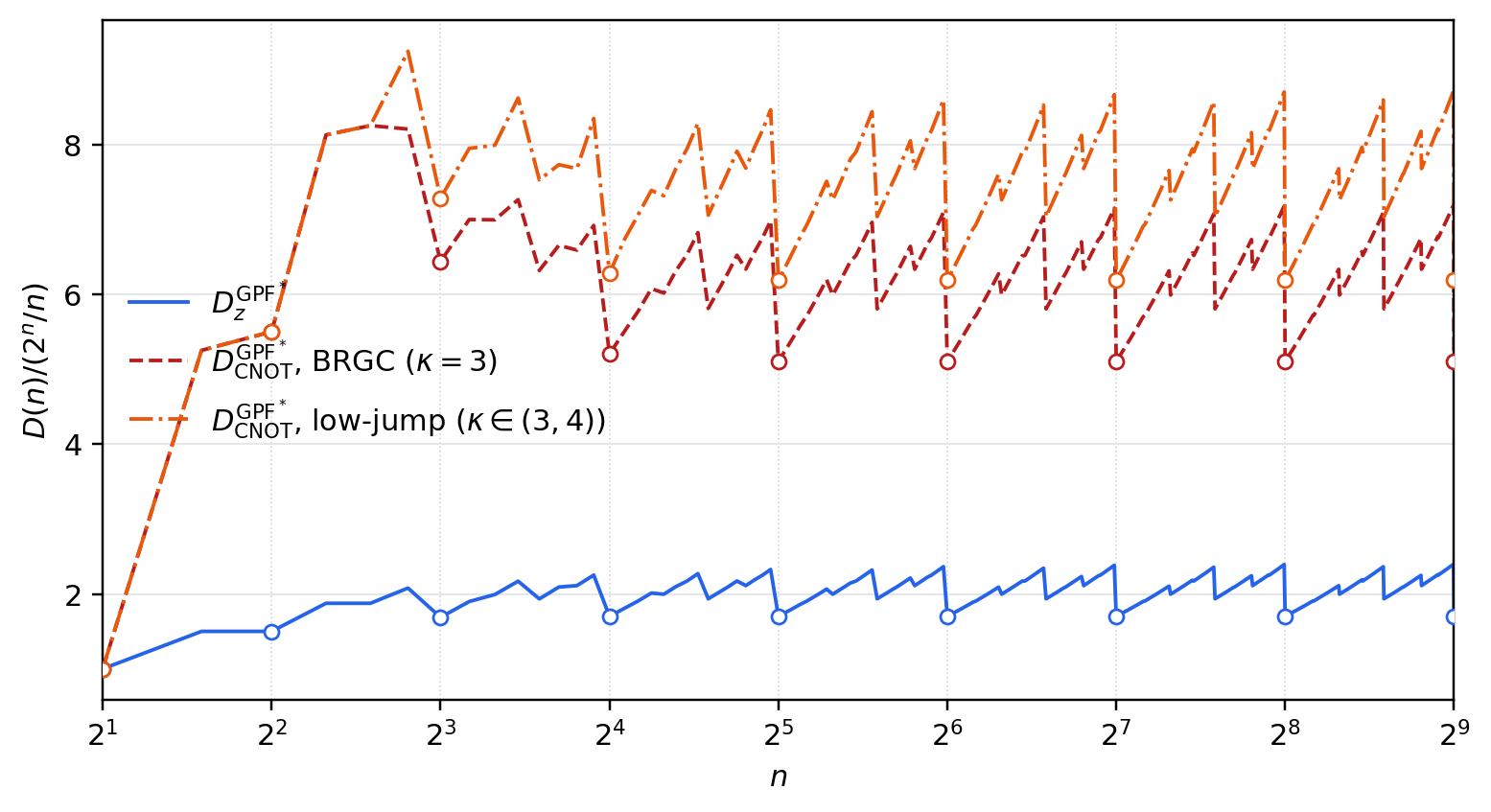}
    \caption{Normalized depth constants for $\GPF^*$. The phase-depth ratio is Gray-code-independent; the CNOT-depth ratio carries the fold factor $\kappa$, converging to $3C/2\approx 5.10$ for the BRGC ($\kappa=3$) and slightly higher for the low-jump sequence ($\kappa\in(3,4)$; Appendix~\ref{app:lowjump}). Hollow markers: $n=2^\lambda$.}
    \label{fig:gpfstar-depth-scaling}
\end{figure}

\begin{theorem}[Closed-form depth constants]
\label{thm:balanced-constants}
Let
\begin{equation}
    C:=\prod_{j=1}^{\infty}\bigl(1+2^{\,1-2^{j-1}}\bigr)
    =2\cdot\tfrac{3}{2}\cdot\tfrac{9}{8}\cdot\tfrac{129}{128}\cdots
    \approx 3.40147 .
    \label{eq:const-C}
\end{equation}
Under all-to-all connectivity, the GPF construction synthesizes arbitrary $n$-qubit diagonal unitaries without ancillas with
\begin{equation}
    D^{\GPF}_z(n)=O(2^n/n),\qquad
    D^{\GPF}_{\CNOT}(n)=O(2^n/n),
\end{equation}
and the leading constants are governed entirely by $C$. As $\lambda\to\infty$, the balanced subsequence $n=2^\lambda$ gives
\begin{equation}
    \frac{D^{\GPF}_z(2^\lambda)}{2^n/n},\ \frac{D^{\GPF}_{\CNOT}(2^\lambda)}{2^n/n}\ \to\ C\approx 3.40147 ,
\end{equation}
while the maximally unbalanced subsequence $n=2^\lambda-1$ gives
\begin{equation}
\begin{aligned}
    &\frac{D^{\GPF}_z(2^\lambda-1)}{2^n/n},\ \frac{D^{\GPF}_{\CNOT}(2^\lambda-1)}{2^n/n}\\
    &\qquad\to\ 2(C-1)\approx 4.80294 .
\end{aligned}
\end{equation}
For the phase-prioritized construction $\GPF^*$, the phase-depth limits are Gray-code-independent,
\begin{equation}
\begin{aligned}
    \frac{D^{\GPF^*}_z(2^\lambda)}{2^n/n}&\to\tfrac{C}{2}\approx 1.70074,\\
    \frac{D^{\GPF^*}_z(2^\lambda-1)}{2^n/n}&\to C-1\approx 2.40147,
\end{aligned}
\end{equation}
whereas the CNOT-depth limits carry the fold factor $\kappa$: $D^{\GPF^*}_{\CNOT}(n)/(2^n/n)\to\tfrac{\kappa}{2}C$ on $n=2^\lambda$ and $\to\kappa(C-1)$ on $n=2^\lambda-1$. For the BRGC ($\kappa=3$) these are $\tfrac{3C}{2}\approx 5.10221$ and $3(C-1)\approx 7.20441$; for a low-jump sequence $\kappa\in(3,4)$ is width-dependent with no closed form (Fig.~\ref{fig:gpfstar-depth-scaling}).
\end{theorem}

\begin{proof}[Proof sketch]
On $n=2^\lambda$ the split is exact ($n_H=n_L=2^{\lambda-1}$), and the $\GPF^*$ phase recurrence in~\eqref{eq:recurrences} collapses to the scalar product $D^{\GPF^*}_z(2^\lambda)=\bigl(2^{\,2^{\lambda-1}-1}+1\bigr)D^{\GPF^*}_z(2^{\lambda-1})$. Telescoping from $D^{\GPF^*}_z(1)=1$ and factoring out the dominant power of two gives $D^{\GPF^*}_z(2^\lambda)=2^{\,2^\lambda-1-\lambda}\prod_{j=1}^{\lambda}\bigl(1+2^{\,1-2^{j-1}}\bigr)$, so $D^{\GPF^*}_z(2^\lambda)/(2^n/n)=\tfrac12\prod_{j=1}^{\lambda}\bigl(1+2^{\,1-2^{j-1}}\bigr)\to C/2$; the product converges absolutely because its $j$-th factor differs from $1$ by $2^{\,1-2^{j-1}}$, which decays doubly exponentially. The remaining constants follow from the leading coefficients of the other three lines of~\eqref{eq:recurrences}: $D^{\GPF^*}_{\CNOT}$ carries the fold factor $\kappa$ (Section~\ref{sec:graycode}), equal to $3$ for the BRGC; the $\GPF$ phase recurrence replaces $2^{n_H-1}$ by $2^{n_H}$ and thus doubles the constant, and the recursive remainder terms vanish under normalization. The upper envelope is obtained by the same telescoping along $n=2^\lambda-1$, where $n_H=2^{\lambda-1}$ is a perfect power of two and $n_L=2^{\lambda-1}-1$ recurses within the family; the base constant then shifts from $C$ to $C-1$. Full details, including the envelope derivation, are given in Appendix~\ref{app:constants}.
\end{proof}

\begin{corollary}[Explicit constant for all $n$]
\label{cor:all-n}
The two subsequences in Theorem~\ref{thm:balanced-constants} are the lower and upper envelopes of the normalized depth. Hence, for every $n$,
\begin{equation}
    D^{\GPF}_z(n),\,D^{\GPF}_{\CNOT}(n)\le\bigl(2(C-1)+o(1)\bigr)\frac{2^n}{n},
\end{equation}
with the bound approached along $n=2^\lambda-1$ and improved to the constant $C$ along $n=2^\lambda$. Here $o(1)\to 0$ as $n\to\infty$ (the normalized depth approaches the constant $2(C-1)$ from below), so it is a vanishing correction, not one of the fixed $O(1)$ leading constants ($C$, $2(C-1)$) it accompanies.
\end{corollary}

Both envelopes are recorded in Theorem~\ref{thm:balanced-constants} and plotted in Figs.~\ref{fig:gpf-depth-scaling}--\ref{fig:gpfstar-depth-scaling}. The basic GP construction is useful as a correctness baseline, while GPF exposes the parallelism needed to reach the asymptotically optimal scale. The $\GPF^*$ variant is preferable when phase gates are substantially more expensive than CNOT gates, as in some fault-tolerant cost models \cite{Litinski2019SurfaceCodes}.

\section{From GPF Structure to Two-Row Placement}
\label{sec:placement}

The synthesis half showed that GPF realizes any $n$-qubit diagonal unitary in depth $O(2^n/n)$ under all-to-all connectivity (Section~\ref{sec:depth}). We now compile this construction onto a two-dimensional nearest-neighbor grid and show that its asymptotic complexity is preserved. Since only two-qubit gates are constrained by connectivity, the placement and routing analysis targets only the CNOT skeleton of GPF, and the circuit diagrams retain phase gates only to mark insertion positions. From the synthesis half we inherit only its synthesis-side objects, namely the recursive high/low split, the modules $\GPS$/$\GPS^*$, the transition sequence $(f_t)$, the block-count recurrence (the first line of Eq.~\eqref{eq:recurrences}), and the constant $C\approx 3.40147$ (Eq.~\eqref{eq:const-C}); all routing-side machinery (offset states, the SWAP-A/B ring, $\mu_{\mathrm{jump}}$, the folding map, and the two routing-module theorems) is introduced here.

To design tailored routing for GPF, we must identify what connections it requires and what they mean physically. This section makes the ``connect high and low sub-blocks'' action of Section~\ref{sec:gpf} concrete: every recursive node produces a family of bipartite matchings, and the natural physical placement for them is a two-row grid. The two geometries the rest of the paper analyzes---two-row top level and single-row recursion---are summarized by one recursion model (Fig.~\ref{fig:recursion-tree}).

\subsection{The Common Connectivity Requirement: Bipartite Matching}
\label{sec:bipartite}

At the top level, the high--low mixing is performed by GPS. GPS parallelizes multiple GP tracks: each low-space qubit is a phase-injection track, and the high-space qubits supply control sources; when step $t$ of the high Gray traversal flips coordinate $f_t$, each low track interacts with the cyclically shifted high qubit. Under all-to-all connectivity this step is a set of disjoint, parallelizable CNOTs (the GPS CNOT skeleton of Fig.~\ref{fig:gps}).

Collecting all the high--low connections a single GPS may use across all steps, one sees that it requires a family of cyclic-offset matchings between the high set $\mathcal{Q}_H$ and the low set $\mathcal{Q}_L$: each step is a perfect matching under a certain offset, and all steps together cover connections approaching a complete bipartite graph (Fig.~\ref{fig:gps-bipartite}).

\begin{figure}[!t]
    \centering
    \includegraphics[width=\columnwidth]{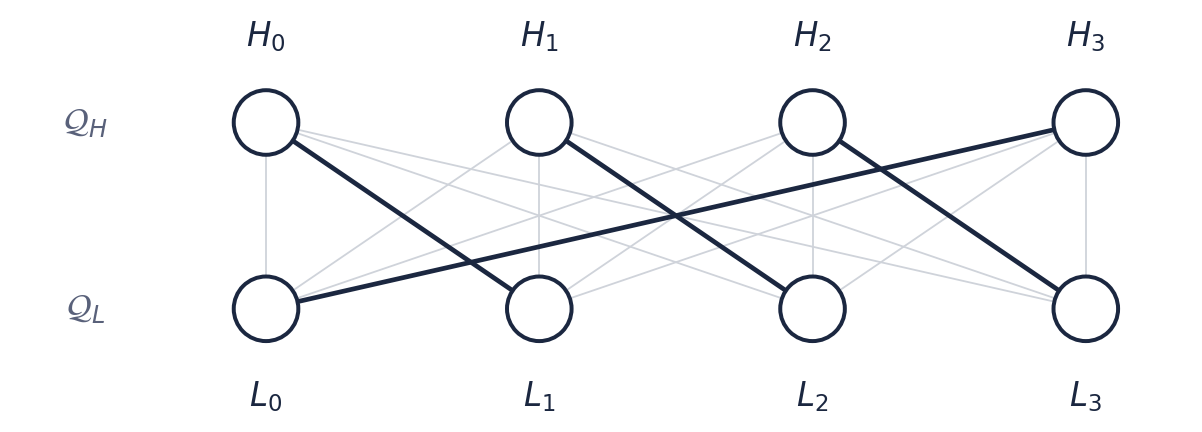}
    \caption{Bipartite connectivity abstraction of GPS. Gray edges represent cross-group connections the top-level GPS may access across different steps; black edges illustrate one specific cyclic-offset matching.}
    \label{fig:gps-bipartite}
\end{figure}

This bipartite matching is not exclusive to the top level. The GPS and GPS$^*$ at recursive nodes face the same type of local problem---splitting a sub-segment into high and low sub-blocks and organizing a family of matchings between them according to their Gray-path schedules. In other words, GPS and GPS$^*$ share the same connectivity type at the synthesis level---multi-layer bipartite matching between two logical halves. The two differ only in the order and multiplicity of matching visits---GPS visits layer by layer in transition-sequence order, GPS$^*$ alternates forward and reverse to trade more CNOT layers for fewer phase layers---not in the geometric type of connectivity.

The core problem for restricted-topology compilation thus becomes concrete: all-to-all connectivity can directly execute any offset matching; nearest-neighbor connectivity must first bring the high and low qubits that need to interact to physically adjacent positions before each matching layer. This converts the problem from ``which CNOTs to synthesize'' to ``how to place qubits and move them between matchings''---placement and routing.

\subsection{Two-Row Embedding and Notation Hierarchy}
\label{sec:tworow-embedding}

The bipartite connectivity almost immediately dictates the placement: put the high half on the upper row and the low half on the lower row, so that each column's vertical edge naturally carries one high--low interaction. The number of columns equals the larger half-space size,
\begin{equation}
    n_H=\lceil n/2\rceil.
\end{equation}
When $n$ is even, $n=2n_H$; when $n$ is odd, $n_H=n_L+1$, and an idle physical site can be appended to the shorter low row to equalize both rows to length $n_H$; this idle site serves only as a routing placeholder and carries no logical qubit. The embedding thus consists of two paths of length $n_H$ plus $n_H$ vertical column edges: horizontal edges execute intra-row SWAPs, and vertical edges at the appropriate time execute cross-row CNOTs (Fig.~\ref{fig:initial}).

\begin{figure}[!t]
    \centering
    \includegraphics[width=\columnwidth]{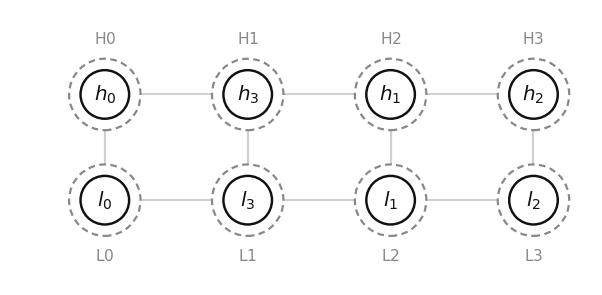}
    \caption{Two-row embedding for $n=2\times 4$: the high half on the upper row, the low half on the lower row (sites $\mathrm{H}_i/\mathrm{L}_i$, occupancy labels $h_i/l_i$; Table~\ref{tab:notation}). Each column's vertical edge carries one high--low interaction; horizontal edges are intra-row SWAPs.}
    \label{fig:initial}
\end{figure}

On the compilation side the primary $\mathrm{H}_i/\mathrm{L}_i$ identity of Section~\ref{sec:gpf} is reused as the fixed physical site, while routing additionally assigns a per-phase occupancy label $h_i/l_i$; together with the global logical index $x_i$ and the recursive coordinate $Q_i$ these make up the four notation layers:

\begin{table}[!t]
    \centering
    \caption{Notation hierarchy for qubit labels. The primary identity $\mathrm{H}_i/\mathrm{L}_i$ is introduced at the synthesis-side high/low partition (Section~\ref{sec:gpf}) and reused as the physical sites on the two-row grid; the remaining labels are sub-algorithm-specific relabelings.}
    \label{tab:notation}
    \begin{tabular}{@{}lll@{}}
        \toprule
        Layer & Notation & Scope \\
        \midrule
        Primary qubit/site & $\mathrm{H}_i$, $\mathrm{L}_i$ & Both halves \\
        Global logical & $x_i$ & Synthesis (basis state) \\
        Offset occupancy & $h_i$, $l_i$ & Section~\ref{sec:tworow} only \\
        Recursive & $Q_0,Q_1,\ldots$ & Section~\ref{sec:onerow} only \\
        \bottomrule
    \end{tabular}
\end{table}

Fig.~\ref{fig:initial} shows the physical sites $\mathrm{H}_i/\mathrm{L}_i$ together with the canonical top-level occupancy $h_i/l_i$ (Table~\ref{tab:notation}). The two-row embedding is the common placement premise for both the top-level and recursive problems; the offset states and cross-row matchings belong only to the top-level GPS.

\begin{figure}[!t]
    \centering
    \begin{tikzpicture}[
        every node/.style={align=center,font=\footnotesize},
        edge from parent/.style={draw,-latex,thick},
        level 1/.style={sibling distance=28mm,level distance=20mm}]
      \node[draw,rounded corners,fill=black!4] (root) {$\GPF(n)$\\\emph{recursive split} $\mathcal{Q}_H\sqcup\mathcal{Q}_L$}
        child {node[draw,rounded corners,fill=blue!6,text width=20mm] {(i) top-level GPS\\\textbf{two rows}, $2\times n_H$\\Sec.~\ref{sec:tworow}}}
        child {node[draw,rounded corners,fill=red!6,text width=20mm] {(ii) high half\\$\GPF(n_H)$\\\textbf{single row}, Sec.~\ref{sec:onerow}}}
        child {node[draw,rounded corners,fill=red!6,text width=20mm] {(iii) low half\\$\GPF^*(n_L)$\\\textbf{single row}, Sec.~\ref{sec:onerow}}};
    \end{tikzpicture}
    \caption{Unified recursion model. $\GPF(n)$ is (i)~a set of top-level GPS-2L mixings between the two halves on a $2\times n_H$ grid, plus (ii)~a high-half recursion $\GPF(n_H)$ and (iii)~a low-half recursion $\GPF^*(n_L)$, both on \emph{single-row} segments. Two-row geometry occurs only at the top level; after the first split, all deeper blocks are single-row (with higher-level qubits acting as bystanders on the line). This figure is the mental model referenced throughout Sections~\ref{sec:routing}--\ref{sec:full}.}
    \label{fig:recursion-tree}
\end{figure}
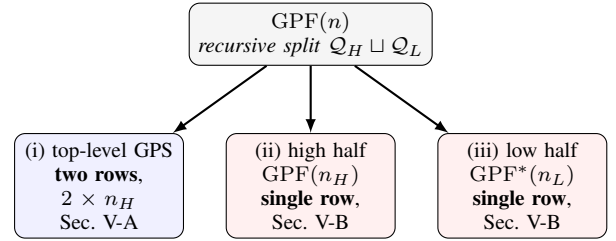

The compilation is now split into two isomorphic subproblems, differing only in geometry: the top-level GPS routes its cyclic-offset matchings on a two-row grid (Section~\ref{sec:tworow}), and the recursive GPS/GPS$^*$ routes its local matchings on single-row segments (Section~\ref{sec:onerow}). Both solve the same task---implementing a sequence of bipartite matchings with nearest-neighbor gates layer by layer---and the next two sections give the closed-form cost under each geometry. The single-row geometry is the more constrained of the two; folding the top-level block onto a single row as well would place the entire construction on one line, and the two-row top level is retained only as the cost-optimized choice.

\section{Routing Modules}
\label{sec:routing}

\subsection{Top-Level GPS: Two-Row Routing (Theorem A)}
\label{sec:tworow}

This section routes a top-level GPS block on the two-row grid, encoding its matchings as offset states and executing them by moving along a SWAP-A/B ring. Reading off the resulting per-block cost gives the first routing-module theorem (Theorem~\ref{thm:gps2l}): a closed form whose only size-dependent quantity is a single constant of the transition sequence, the cyclic jump density $\mu_{\mathrm{jump}}$. The argument proceeds in four steps: encode matchings as offset states $s_d$; move between states with two parallel-swap layers whose cost equals a ring distance; accumulate the ring distances of the transition sequence into the closed-form cost $J$ (Theorem~\ref{thm:gps2l}); and pin $\mu_{\mathrm{jump}}$ near its lower bound with the low-jump sequence.

Following the cost notation of Section~\ref{sec:preliminaries}, the subscript $\mathrm{route}$ marks routing-added CNOT cost---``incurred by routing,'' as opposed to the circuit's own logical CNOTs---and here the superscript names the routed object as a module--geometry pair: one of GPS, GPS$^*$, GPF, GPF$^*$ combined with $2\mathrm{L}$ (two-row) or $1\mathrm{L}$ (single-row), kept explicit throughout. In particular, $\GPS\text{-}2\mathrm{L}$ is the two-row top-level GPS (its folded version is $\GPS\text{-}1\mathrm{L}$).

The object to be routed (Appendix~\ref{app:examples}, Fig.~\ref{fig:logical-n8}) is the logical top-level GPS for $n=8$ ($n_H=4$), where high qubits $\mathrm{H}_0,\ldots,\mathrm{H}_3$ and low qubits $\mathrm{L}_0,\ldots,\mathrm{L}_3$ are connected by cross-block CNOTs (control on high, target on low) whose spans vary with the transition sequence and are mostly non-adjacent. The task is to implement all of them on the $2\times 4$ nearest-neighbor grid with controllable, closed-form cost.

\subsubsection{Encoding Matchings as Offset States}
\label{sec:offset}

The idea: rather than tracking ``which two qubits interact'' per pair, introduce a unified state label for an entire matching layer, so that ``executing a matching layer'' becomes ``moving the layout to a state.''

For the ring we relabel the high/low qubits $\mathrm{H}_i/\mathrm{L}_i$ as offset occupancy tokens $h_i,l_i$ ($i\in\mathbb{Z}_{n_H}$), following the top-level GPF partition. If a padding site exists, the corresponding $l_i$ is an idle physical site. These $h_i/l_i$ are neither the global logical index $x_i$ nor the fixed physical sites $\mathrm{H}_i/\mathrm{L}_i$, but merely mark ``which bipartite matching currently falls on physical edges.'' The canonical occupancy for $n=2\times 4$ is the one shown in Fig.~\ref{fig:initial}. We fix the \emph{canonical arrangement} in which column $j$ carries the labels of index $\sigma(j)$ on both rows, where $\sigma(2k)=k$ and $\sigma(2k+1)=n_H-1-k$ (the inward-interleaved order $0,n_H{-}1,1,n_H{-}2,\dots$); this canonical arrangement is the offset state $s_0$ (all $\Delta_i=0$), and the offset ring is generated from it.

Let column $i$ have $h_{\pi_H(i)}$ on the upper physical site $\mathrm{H}_i$ and $l_{\pi_L(i)}$ on the lower physical site $\mathrm{L}_i$, where $\pi_H,\pi_L$ are the current occupancy maps. Define the column offset
\begin{equation}
    \Delta_i=(\pi_L(i)-\pi_H(i))\bmod n_H.
\end{equation}
If all non-idle columns satisfy $\Delta_i=d$, we say the layout is in offset state $s_d$. In this state, each vertical column edge $\mathrm{H}_i\leftrightarrow\mathrm{L}_i$ carries exactly one pair $(h_a,l_{a+d})$, and all valid columns simultaneously realize the cyclic-offset matching
\begin{equation}
    M_d=\{(h_a,l_{a+d})\mid a\in\mathbb{Z}_{n_H},\;l_{a+d}\text{ valid}\}.
\end{equation}
This is the entire encoding: reaching $s_d$ makes all cross-row CNOTs of matching $M_d$ fall in parallel on the column edges. Each matching layer required by the top-level GPS takes the form of some $M_d$, so ``executing a matching layer'' translates to ``routing to the corresponding offset state.''

\begin{lemma}[Offset states yield GPS cross-group matchings]
\label{lem:offset}
If the layout is in $s_d$, then the vertical column edges simultaneously realize the valid matching $M_d$ of the top-level GPS at offset~$d$.
\end{lemma}

\begin{proof}
By the definition of $s_d$, every valid column $i$ satisfies $\pi_L(i)-\pi_H(i)\equiv d\pmod{n_H}$. Let the high qubit at column $i$ be $h_a$ ($a=\pi_H(i)$); then the same column's low qubit is $l_{a+d}$. Iterating over all valid columns yields $M_d$; columns landing on idle sites contribute no logical CNOT.
\end{proof}

\subsubsection{The Offset-State Ring: SWAP-A/SWAP-B and Movement Cost}
\label{sec:ring}

Since ``executing a matching'' is equivalent to ``reaching a state,'' the remaining question is how to move between states and at what cost.

We use two types of parallel adjacent swap layers. SWAP-A simultaneously swaps upper-row pairs $(0,1),(2,3),\dots$ and lower-row pairs $(1,2),(3,4),\dots$. SWAP-B uses the complementary adjacent edges: upper-row $(1,2),(3,4),\dots$ and lower-row $(0,1),(2,3),\dots$. Each type uses only disjoint nearest-neighbor edges, so each counts as one SWAP depth layer with exactly $n_H-1$ parallel SWAP gates.

\begin{lemma}[Offset ring]
\label{lem:offset-ring}
From the canonical arrangement $s_0$, each SWAP-A or SWAP-B layer maps a uniform offset state $s_d$ to a uniform offset state $s_{d+1}$ or $s_{d-1}$ (indices mod $n_H$), and the alternating sequence SWAP-A, SWAP-B, SWAP-A, $\dots$ realizes $s_0\to s_1\to\cdots\to s_{n_H-1}\to s_0$, visiting every offset state exactly once. Hence the offset states form a single cycle of length $n_H$ under the SWAP-A/B primitive, and the minimum number of SWAP layers taking $s_a$ to $s_b$ is $\Dist_{n_H}(a,b)$.
\end{lemma}

The interleaving is what makes the ring uniform: in folded-cycle coordinates the even- and odd-brick layers act as two reflections of $\mathbb{Z}_{n_H}$ whose composite is a unit rotation, so applying complementary bricks to the two rows shifts the common offset by one. The full argument is given in Appendix~\ref{app:offset-ring}. Fig.~\ref{fig:delta-cycle} illustrates this ring for $n_H=4$; it is the physical primitive for scheduling, not the GPS visit order itself.

\begin{figure*}[!t]
    \centering
    \includegraphics[width=0.92\textwidth]{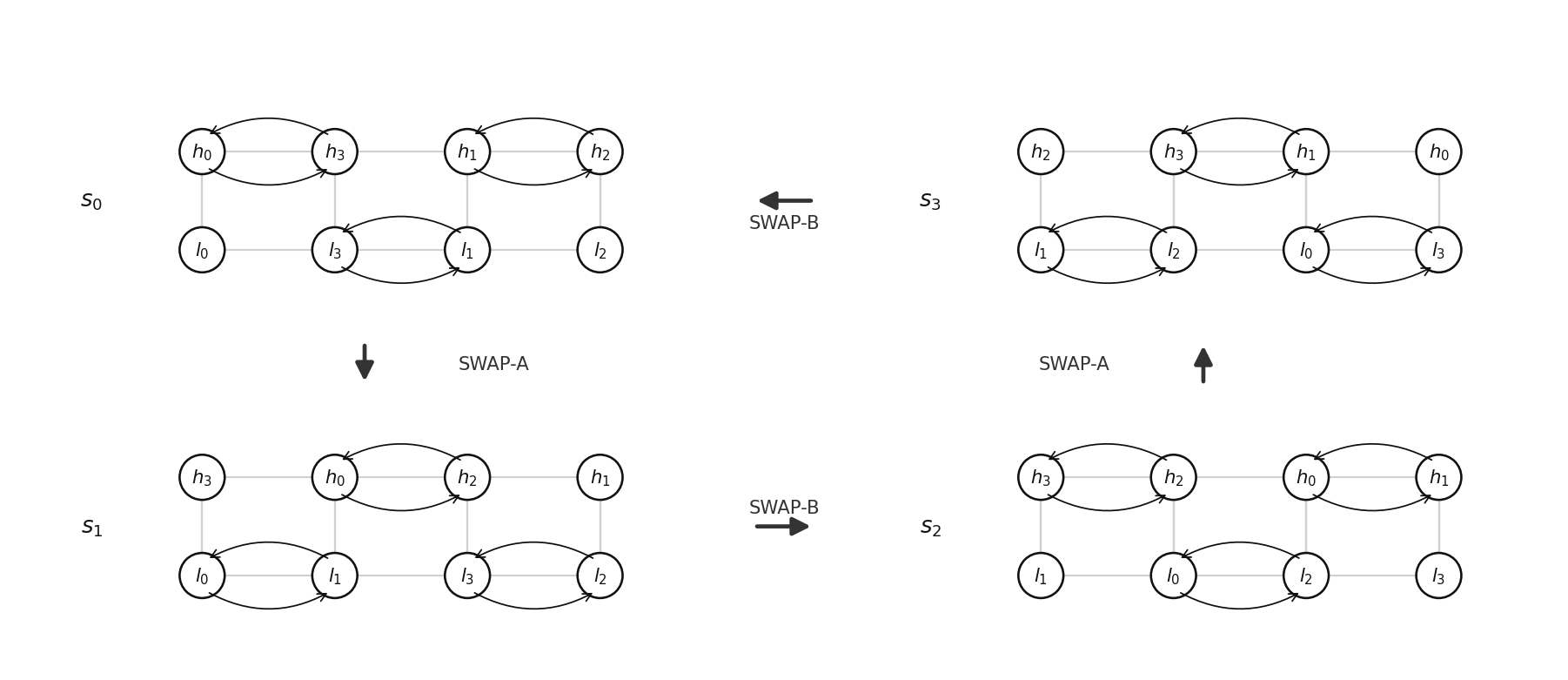}
    \caption{Base offset ring for $n_H=4$. Each cell shows one offset state together with the SWAP layer that advances it to the next state. Alternating SWAP-A/B cycles $s_0\to s_1\to s_2\to s_3\to s_0$, reaching every offset state.}
    \label{fig:delta-cycle}
\end{figure*}

The movement cost between states is the ring distance, the cyclic distance $\Dist_{n_H}$ of Section~\ref{sec:graycode}. The shortest SWAP depth from $s_a$ to $s_b$ is $\Dist_{n_H}(a,b)$, with SWAP gate count $(n_H-1)\,\Dist_{n_H}(a,b)$. The ring distance gives the depth, and multiplying by the per-layer parallelism gives the gate count; this correspondence is the source of all cost formulas in this section.

Consecutive same-type swap layers arise as optimal behavior, not a defect. From $s_d$, applying SWAP-A shifts the offset by $+1$ or $-1$ depending on the parity of $d$. When routing reverses direction on the ring (e.g., $s_0\to s_1\to s_0$), two consecutive same-type layers (A,A or B,B) naturally appear, a direct consequence of moving to the nearest target along $\Dist_{n_H}$ without resetting to $s_0$ after every matching, which avoids unnecessary round trips.

\subsubsection{Transition Sequences and the Closed-Form Cost}
\label{sec:jump-cost}

The top-level GPS visits offset states not in natural order but in the order given by the synthesis-side transition sequence $G=(f_0,f_1,\ldots,f_{2^{n_H}-1})$ of Section~\ref{sec:graycode}, a valid closed cycle on the $n_H$ high coordinates. The top-level GPS visits $s_{f_0},s_{f_1},\ldots,s_{f_{2^{n_H}-1}}$ in order, then closes back to the start. This is the connection point between synthesis and compilation: the Gray traversal gives $(f_t)$, the transition sequence determines which offset states to visit, and the ring distances between adjacent states determine the SWAP layer count.

Accumulating the movement cost of Section~\ref{sec:ring} along this visit sequence gives the internal SWAP overhead of one top-level GPS block. From $s_{f_t}$ to $s_{f_{t+1}}$ requires $\Dist_{n_H}(f_t,f_{t+1})$ layers, so the internal state-switching depth of the full closed cycle is
\begin{equation}
    J=\sum_{t=0}^{2^{n_H}-1}\Dist_{n_H}\!\left(f_t,\,f_{(t+1)\bmod 2^{n_H}}\right),
    \label{eq:jump-cost}
\end{equation}
the cyclic jump cost of the transition sequence---the jump density of Eq.~\eqref{eq:mu-def} scaled by the cycle length, $J=\mu_{\mathrm{jump}}(G)\,2^{n_H}$---with internal SWAP gate count $(n_H-1)J$.

\begin{lemma}[Jump cost gives internal SWAP overhead]
\label{lem:jump-cost}
Among offset-ring (SWAP-transport) schedules, executing the full transition sequence from $s_{f_0}$ and closing back to $s_{f_0}$ requires internal state-switching SWAP depth at least $J$, attained by routing each step along the ring, with SWAP gate count $(n_H-1)J$.
\end{lemma}

\begin{proof}
By Lemma~\ref{lem:offset-ring} the offset states form a cycle on which moving from $s_{f_t}$ to $s_{f_{t+1}}$ takes at least $\Dist_{n_H}(f_t,f_{t+1})$ SWAP layers, attained along the ring; the closed cycle completes all switches (including the last-to-first return), and these minimum layer counts add to $J$. Each layer contains $n_H-1$ parallel SWAPs, so the gate count is $(n_H-1)J$. The minimum is over offset-ring transport; the alternative remote-CNOT realization implements the same matchings without transport and is accounted separately.
\end{proof}

When converting layers to overhead, $1$ SWAP decomposes into $3$ CNOTs (depth $3$), so we use routing-added CNOTs as the uniform unit, which yields the section's main result.

\begin{theorem}[Routing cost of a top-level GPS-2L block]
\label{thm:gps2l}
On a $2\times n_H$ two-row grid using nearest-neighbor SWAPs, routing one top-level GPS block incurs routing-added CNOT depth and gate count
\begin{equation}
\begin{aligned}
    D^{\GPS\text{-}2\mathrm{L}}_{\mathrm{route}}(n)
    &=3J=3\mu_{\mathrm{jump}}\,2^{n_H},\\
    S^{\GPS\text{-}2\mathrm{L}}_{\mathrm{route}}(n)
    &=3(n_H{-}1)J=3(n_H{-}1)\mu_{\mathrm{jump}}\,2^{n_H}.
\end{aligned}
\end{equation}
\end{theorem}

\begin{proof}
By Lemma~\ref{lem:jump-cost}, the transition-sequence closed cycle has internal state-switching depth $J$ and SWAP gate count $(n_H-1)J$; on a restricted topology, $1$ SWAP decomposes into $3$ CNOTs. Substituting $J=\mu_{\mathrm{jump}}\,2^{n_H}$ (Eq.~\eqref{eq:mu-def}) yields the result.
\end{proof}

The closed form is exact and completely decouples the routing cost from the choice of transition sequence: all size dependence is absorbed into the single constant $\mu_{\mathrm{jump}}$. Since adjacent transition positions must differ and each step has ring distance $\ge 1$, there is always the lower bound $\mu_{\mathrm{jump}}\ge 1$.

The closed form measures the internal cost of the closed cycle. An executable block additionally carries an $O(1)$ boundary correction $b$: the entry from the fixed initial arrangement $s_0$ to the first visited offset $s_{f_0}$, minus the loop-closing return that need not be materialized. The executable block has $J+b$ SWAP layers with $b=O(1)$. At $n=8$, for instance, $n_H=4$ and low-jump give $J=16$, while the executable routed circuit of Fig.~\ref{fig:routed-n8} has 15 SWAP layers and 45 SWAP gates, i.e.\ $b=-1$, $D^{\GPS\text{-}2\mathrm{L}}_{\mathrm{route}}=45$ CNOT depth and $S^{\GPS\text{-}2\mathrm{L}}_{\mathrm{route}}=135$ routing CNOTs. This $O(1)$ correction does not change the leading-order scaling, so the analysis uses Theorem~\ref{thm:gps2l}. The same transition sequence independently affects the synthesis-side depth through its own constant; the two are decoupled, and the compilation side substitutes only $\mu_{\mathrm{jump}}$.

Fig.~\ref{fig:routed-n8} shows the routing result for $n=8$: the same circuit as Fig.~\ref{fig:logical-n8}, but every cross-block CNOT is preceded by SWAP-A/SWAP-B layers that bring the matching to vertical nearest neighbors before execution; all two-qubit gates now fall on coupling edges of the $2\times 4$ grid. A 2.5D rendering of the same schedule, convenient for reading the two-row physical topology, is given in Appendix~\ref{app:examples}, Fig.~\ref{fig:routed-3d-n8}.

\begin{figure*}[!t]
    \centering
    \includegraphics[width=0.88\textwidth]{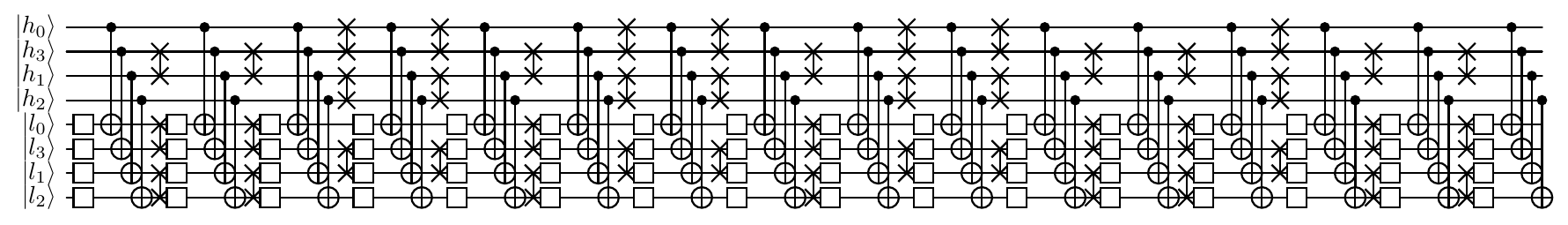}
    \caption{Routed top-level GPS circuit ($n=8$). Boxes abbreviate the inserted $R_z(\varphi_s)$ phase gates.}
    \label{fig:routed-n8}
\end{figure*}

\subsubsection{The low-jump Transition Sequence}
\label{sec:minjump}

Theorem~\ref{thm:gps2l} reduces top-level routing optimization to a pure combinatorial problem: over all valid closed-cycle transition sequences, minimize the jump density $\mu_{\mathrm{jump}}$ of Eq.~\eqref{eq:mu-def}. We adopt a \emph{low-jump} sequence that keeps $\mu_{\mathrm{jump}}$ small (Section~\ref{sec:graycode}); its key property is summarized below. It is not known to be globally optimal---the name is descriptive.

\begin{lemma}[low-jump attains the cost lower bound for small $n_H$]
\label{lem:minjump}
For the low-jump transition sequence,
\begin{itemize}
    \item $n_H\le 5$: $\mu_{\mathrm{jump}}=1.0$, attaining the hard lower bound $J=2^{n_H}$ and proved globally optimal by iterative-deepening exact search;
    \item $n_H=6$: $\mu_{\mathrm{jump}}=66/64\approx 1.03$, exhaustively proved that no smaller valid closed cycle exists;
    \item $n_H\ge 7$: a recursive construction closes and $\mu_{\mathrm{jump}}$ converges to $\approx 1.06$, reported as experimental input without a full-scale optimality claim.
\end{itemize}
\end{lemma}

The sequences satisfy a closed-cycle \emph{validity} constraint (a genuine Gray cycle traversing all $2^{n_H}$ modes and returning to mode $0$, the same criterion as on the synthesis side), and are obtained by exact search for $n_H\le 6$ and a deterministic recursive construction for $n_H\ge 7$. Both the validity constraint and the construction are detailed in Appendix~\ref{app:lowjump}.

Under low-jump, when $n$ is even ($n_H=n/2$), Theorem~\ref{thm:gps2l} becomes
\begin{equation}
\begin{aligned}
    D^{\GPS\text{-}2\mathrm{L}}_{\mathrm{route}}(n)
    &=3\mu_{\mathrm{jump}}\,2^{n/2},\\
    S^{\GPS\text{-}2\mathrm{L}}_{\mathrm{route}}(n)
    &=3\mu_{\mathrm{jump}}\bigl(\tfrac{n}{2}-1\bigr)2^{n/2},
\end{aligned}
\end{equation}
exact in $\mu_{\mathrm{jump}}$: all size dependence is in $2^{n/2}$, the only constant being $\mu_{\mathrm{jump}}$. For $n_H\ge 7$ the low-jump value is supplied by a recursive construction that is not proven optimal (Appendix~\ref{app:lowjump}); we compute $\mu_{\mathrm{jump}}\in[1.058,1.063]$ and report $\mu_{\mathrm{jump}}\approx 1.06$ as an empirical upper bound, giving $D^{\GPS\text{-}2\mathrm{L}}_{\mathrm{route}}\lesssim 3.18\cdot 2^{n/2}$. For a concrete finite $n$ one uses the exact per-block $\mu_{\mathrm{jump}}$, which is $1.0$ whenever a block's high width is at most $5$; the $n=8$ instance therefore uses $\mu_{\mathrm{jump}}=1.0$ and gives normalized depth $12.0$ rather than $12.7$, and the value $\approx 1.06$ enters only the large-$n$ asymptotic constant. This completes Theorem~A.

\subsection{Recursive GPS/GPS$^*$: Single-Row Routing (Theorem B)}
\label{sec:onerow}

This section folds the same GPS block onto a single row. The single-row GPS/GPS$^*$ is the same circuit as the two-row GPS---the same low-jump offset ring, the same SWAP-A/B and cross-block CNOT layers---so it is a specialization of the two-row construction rather than a fresh problem (Sections~\ref{sec:tworow} and \ref{sec:onerow} are one routing problem under two geometries; Fig.~\ref{fig:recursion-tree}); only two things change: the mapping---which lane each qubit occupies on the line---and the routing---layers that were parallel nearest-neighbor on two rows now contain remote gates to be decomposed to nearest-neighbor on one line. Reading off its cost gives the second routing-module theorem (Theorem~\ref{thm:gps1l}): exactly twice the two-row cost of Theorem~\ref{thm:gps2l} at the top recursion level, with a closed form at every deeper level.

The recursive portion is not only GPS$^*$: the low-driven branch uses GPS$^*$, but the $s_L=0$ high branch in standard GPF remains GPS. As noted in Section~\ref{sec:bipartite}, the two share routing geometry and differ only in synthesis-side visit order, so we use GPS$^*$ as the running example and GPS follows identically. Accordingly we give the mapping, the single-row routing method, examples, and the cost (Theorem~\ref{thm:gps1l}). Upon entering this section, the $h_i/l_i$ cross-row relations of Section~\ref{sec:tworow} are retired; qubits on the single row are numbered only by fixed physical coordinate $Q_0,Q_1,\ldots$. The object to be routed is the logical single-row GPS$^*$ for $n=8$ (Appendix~\ref{app:examples}, Fig.~\ref{fig:star-logical-n8}), whose cross-block CNOT spans are highly irregular.

\subsubsection{Folding Map: Inward-Convergent Recursive Selection}
\label{sec:folding}

The two-row layout is folded onto a single line: adjacent physical positions $(2c,2c+1)$ still form a ``column,'' so that at offset $0$ the high and low qubits are adjacent and the CNOT is nearest-neighbor. Which position holds the high qubit (and which the low) is not fixed by a naive even/odd folding---that would drastically increase lane separation---but by inward-convergent recursive selection.

Recursive selection is a balanced binary recursion: pair the currently active positions on the line by adjacency, and let each pair take its inner-side representative into the next deeper level---the left pair its right element, the right pair its left---halving the active count. Within a block of length $2^r$, the highest-level paired left and right representatives (relative to the block's left end) satisfy
\begin{gather}
    L_r=R_{r-1},\quad R_r=2^{r-1}+L_{r-1},\quad (L_1,R_1)=(0,1)
    \notag\\
    \Longrightarrow\quad
    L_r=\Bigl\lfloor\tfrac{2^r}{3}\Bigr\rfloor,\quad
    R_r=\Bigl\lfloor\tfrac{2^{r+1}}{3}\Bigr\rfloor,
\end{gather}
and their pairing distance forms the distance spectrum $\delta_r=R_r-L_r=\lfloor(2^r+1)/3\rfloor$ for $r=1,2,\ldots$ For example, with $8$ addresses $\delta_r=[1,1,3]$ and the top-level representatives are $(L_3,R_3)=(2,5)$. The full-block folding uses only round~$0$ of recursive selection, which fixes the high/low assignment in each column; deeper rounds $\delta_2,\delta_3,\ldots$ appear only upon further recursion.

\begin{figure}[!t]
    \centering
    \includegraphics[width=\columnwidth]{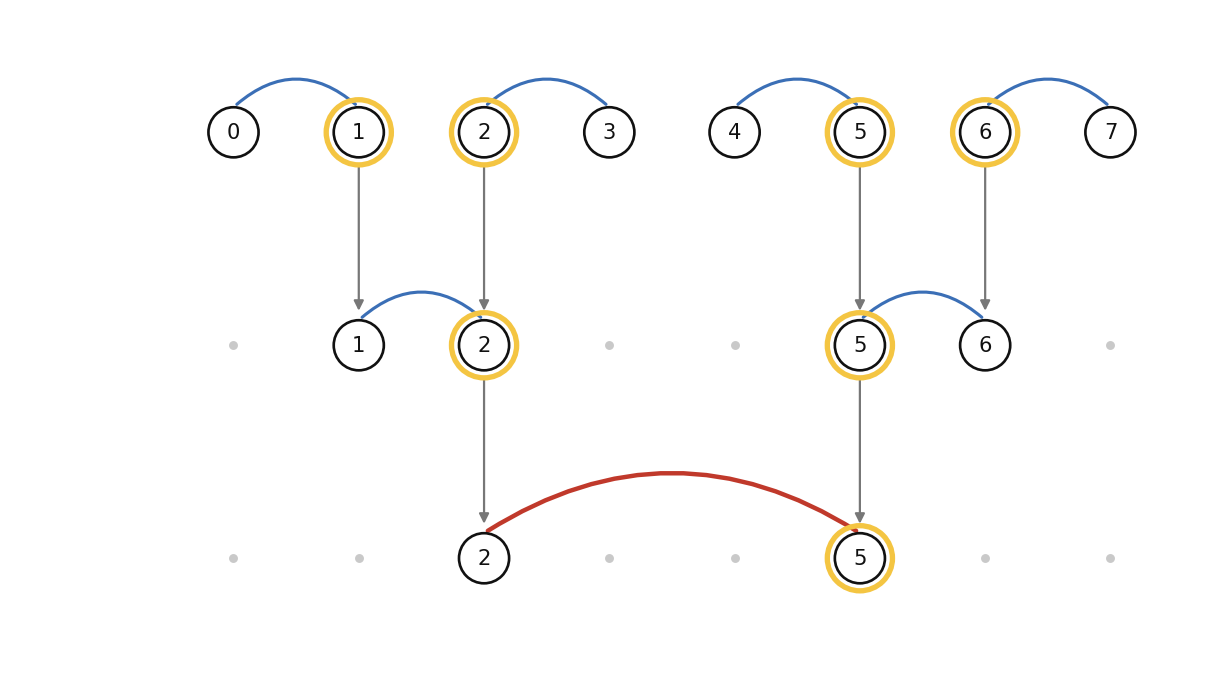}
    \caption{Inward-convergent recursive selection ($K=8$). Adjacent halves take inner-side representatives (gold rings) into the next level: $0$--$7\to\{1,2,5,6\}\to\{2,5\}$; distance spectrum $\delta_r=[1,1,3]$.}
    \label{fig:recursion}
\end{figure}

The reason for using recursive selection rather than naive folding is practical: it minimizes the remote extent of each SWAP layer after folding, compressing the brickwork depth of SWAP-B from $3$ down to $1$ and thereby reducing the total offset-ring routing depth by roughly one-third.

\subsubsection{Single-Row Routing: Intra-Group SWAP-A/B and Inter-Group Remote CNOT}
\label{sec:onerow-routing}

After recursive selection pins each qubit to a fixed lane (Section~\ref{sec:folding}), the folded single-row circuit is Section~\ref{sec:tworow}'s structure transplanted onto a line: a sequence of SWAP-A/B layers interleaved with CNOT layers. Recursive selection partitions the line into two interleaved groups: at the top level, the low group $\{1,2,5,6\}$ and the high group $\{0,3,4,7\}$ for $n=8$. Under this partition the entire GPS/GPS$^*$ block consists of exactly two operations.
\begin{itemize}
    \item \emph{Intra-group SWAP-A/B.} Section~\ref{sec:tworow}'s SWAP-A simultaneously swapped upper-row pairs $(0,1),(2,3),\ldots$ and lower-row pairs $(1,2),(3,4),\ldots$; folded onto a single row, these swaps all fall within the same group. For $n=8$, SWAP-A is the high-group permutation $(0,3),(4,7)$ together with the low-group $(2,5)$, and SWAP-B analogously. Their function is identical to Section~\ref{sec:tworow}: moving the layout between offset states.
    \item \emph{Inter-group remote CNOT.} Each step's cross-block matching connects a high representative to a low representative, a CNOT that crosses the two groups. The groups interleave on the line, so these matchings are mostly remote in the logical layout; for $n=8$, $40$ of the $64$ top-level matchings have span larger than $1$.
\end{itemize}

Single-row routing thus reduces to making both of these nearest-neighbor: intra-group remote SWAPs and inter-group remote CNOTs. Both are present from the top level onward, and recursion merely subdivides each group further, repeating the same decomposition at smaller scale.

Intra-group remote SWAPs decompose into an odd--even brickwork. A SWAP-A/B layer within a group is a permutation, decomposed to nearest-neighbor by an odd--even brickwork. Under the recursive-selection mapping, and independently of $n$, SWAP-A folds to a permutation of maximum distance $3$ with brickwork depth exactly $3$ and gate count $3n_H-1$, while SWAP-B folds to maximum distance $1$ with brickwork depth exactly $1$ and gate count $n_H-1$. Fig.~\ref{fig:swap} shows one SWAP-A layer, remote and decomposed into a three-layer brickwork, and one SWAP-B layer that is already nearest-neighbor.

\begin{figure}[!t]
    \centering
    \begin{subfigure}[b]{0.56\columnwidth}
        \centering
        \includegraphics[height=3.4cm]{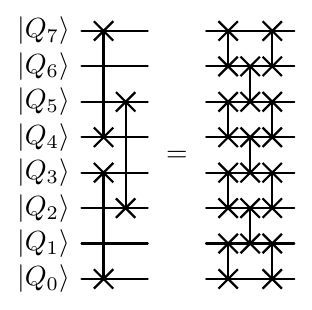}
        \caption{}
        \label{fig:swapA}
    \end{subfigure}
    \hfill
    \begin{subfigure}[b]{0.30\columnwidth}
        \centering
        \includegraphics[height=3.4cm]{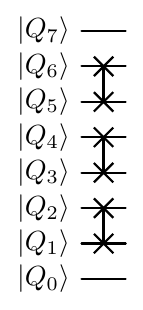}
        \caption{}
        \label{fig:swapB}
    \end{subfigure}
    \caption{The SWAP-A/B layers used by single-row GPS/GPS$^*$ at the first recursion level, with their low-depth nearest-neighbor realization: (a)~SWAP-A folds to a depth-$3$ brickwork, (b)~SWAP-B to depth~$1$. Deeper recursion has the same structure, the SWAPs merely becoming remote and further brickwork-decomposed.}
    \label{fig:swap}
\end{figure}

Under a naive even/odd folding both SWAP types would have brickwork depth $3$; recursive selection compresses SWAP-B to $1$ while SWAP-A remains at $3$ (the one-third saving of Section~\ref{sec:folding}). That SWAP-B is nearest-neighbor is only a feature of the top level and shallow recursion: a few levels deeper, both swap types cross bystanders---qubits from higher levels that remain on the line---increasing distances, and the brickwork handles this uniformly.

Inter-group remote CNOTs decompose into parallel nearest-neighbor ladders. A cross-group CNOT of distance $d$ can be implemented by a pure nearest-neighbor CNOT ladder of $4(d-1)$ gates, with no SWAPs and no ancillas. Recursive selection guarantees that all pairings at the same scale $r$ are equidistant at $\delta_r$ and non-crossing, so these remote CNOTs occupy disjoint intervals on the line and decompose into ladders simultaneously; the layer's depth then depends only on the single distance $\delta_r$, independently of the number of pairings. This is where the distance spectrum $\delta_r$ of Section~\ref{sec:folding} enters on the routing side.

An equivalent shortcut is available: the intra-group SWAP-A/B, while shifting the offset, simultaneously bring each step's matching representatives to adjacent columns, so the matching CNOTs in the routed circuits are directly nearest-neighbor; the remote CNOT ladder above is the alternative that implements the same cross-group CNOT without relying on this transport. This completes single-row routing.

\subsubsection{Examples}
\label{sec:examples}

We apply the method to specific $n=8$ routed circuits (the logical starting point is Fig.~\ref{fig:star-logical-n8}); each realizes its original module.

Example~1, the top-level GPS (Fig.~\ref{fig:gps-1row}): every offset-ring step's intra-group SWAP-A/B is decomposed to nearest-neighbor brickwork, and inter-group matchings fall on adjacent columns. Two further examples are deferred to Appendix~\ref{app:examples}. Example~2 is the top-level GPS$^*$, with the same folding and brickwork and only the synthesis-side schedule switched to the forward--reverse alternation of GPS$^*$ (Fig.~\ref{fig:star-1row}). Example~3 is a one-level recursive GPS$^*$ entering the low group $\{1,2,5,6\}$ with $\{0,3,4,7\}$ as bystanders, where the offset-shifting SWAP$(2,5)$ now crosses the bystanders $\{3,4\}$ and is brickwork-decomposed (Fig.~\ref{fig:segment-1256}); this illustrates deeper-recursion intra-group SWAPs also becoming remote.

\begin{figure*}[!t]
    \centering
    \includegraphics[width=0.88\textwidth]{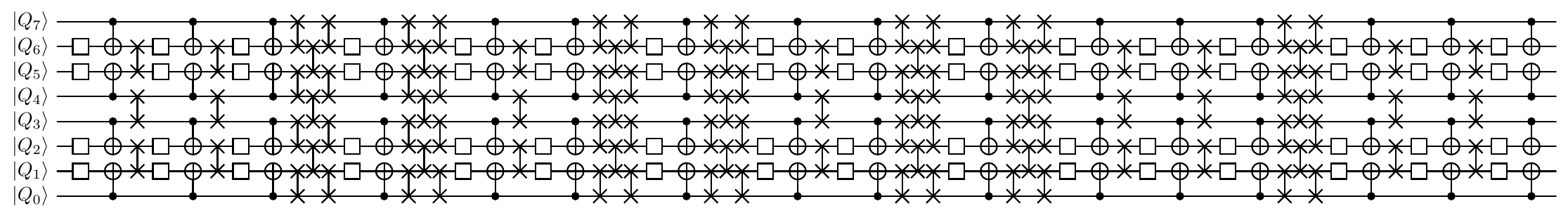}
    \caption{Routed single-row GPS module ($n=8$).}
    \label{fig:gps-1row}
\end{figure*}

\subsubsection{Cost and Complexity}
\label{sec:onerow-cost}

A GPS-1L node at level $r$ of the GPF recursion tree ($r=0$ is the top level) acts on $K_r=n/2^r$ positions, with its high sub-group containing $w_r=n/2^{r+1}$ qubits and its offset ring containing $J_r=\mu_{\mathrm{jump}}\,2^{w_r}$ SWAP-A/B layers. As $r$ increases, the offset ring shortens while each brickwork layer deepens: bystander spacing right-shifts the distance spectrum by $r$ positions, so the SWAP-A layer brickwork depth is exactly $\delta_{r+3}$ and the SWAP-B layer depth is $\delta_{r+2}$ ($\delta_k=\lfloor(2^k+1)/3\rfloor$, Section~\ref{sec:folding}). This deepening is not an empirical input: it follows from the distance spectrum, and its summed contribution is governed by the identity $\delta_{k+1}+\delta_k=2^k$ used in the proof below.

\begin{theorem}[Routing cost of a recursive GPS-1L block]
\label{thm:gps1l}
The routing-added CNOT depth of a GPS-1L block at level $r$ of the GPF recursion tree (high width $w_r=n/2^{r+1}$) is
\begin{equation}
    D^{\GPS\text{-}1\mathrm{L}}_{\mathrm{route}}(n,r)
    =6\,\mu_{\mathrm{jump}}\;2^{\,r}\;2^{\,n/2^{\,r+1}}
\label{eq:gps1l-depth}
\end{equation}
The gate-count leading term (at $r=0$) is $S^{\GPS\text{-}1\mathrm{L}}_{\mathrm{route}}(n,0)=3(n-1)\mu_{\mathrm{jump}}\,2^{n/2}$.
\end{theorem}

\begin{proof}[Proof (depth)]
By the half-and-half property---SWAP-A pushes even offset states to odd, SWAP-B pushes odd back to even (Lemma~\ref{lem:offset-ring}), and the closed cycle from $s_0$ must return to $s_0$---the $J_r$ layers split evenly, $J_r/2$ each for A and B. Summing each layer's depth, multiplying by $3$ for CNOT conversion, and using the identity $\delta_{k+1}+\delta_k=2^k$ (hence $\delta_{r+3}+\delta_{r+2}=2^{r+2}$), gives Eq.~\eqref{eq:gps1l-depth}.

\emph{(Gate count.)} At the dominant top-level block ($r=0$), each SWAP-A layer has $3n_H-1$ gates and each SWAP-B layer has $n_H-1$ gates, so
\begin{equation}
\begin{aligned}
    S^{\GPS\text{-}1\mathrm{L}}_{\mathrm{route}}(n,0)
    &=3\!\cdot\!\tfrac{J}{2}\bigl[(3n_H{-}1){+}(n_H{-}1)\bigr]\\
    &=3(n{-}1)\mu_{\mathrm{jump}}\,2^{n/2}.
\end{aligned}
\end{equation}
\end{proof}

The form parallels Section~\ref{sec:tworow}'s $S^{\GPS\text{-}2\mathrm{L}}_{\mathrm{route}}=3(n_H-1)\mu_{\mathrm{jump}}2^{n_H}$; the width factor changes from $n_H-1$ to $n-1$ (the cost of not parallelizing across two rows). Deeper-level gate counts increase slightly due to brickwork overlap but decay double-exponentially with $2^{n/2^{r+1}}$ and are sub-leading. Both formulas have $\mu_{\mathrm{jump}}$ as the only constant, and executable single blocks each carry an $O(1)$ boundary correction. At $r=0$, $D^{\GPS\text{-}1\mathrm{L}}_{\mathrm{route}}(n,0)=2\,D^{\GPS\text{-}2\mathrm{L}}_{\mathrm{route}}(n)$: folding to a single row is exactly twice the two-row cost, establishing Theorem~B. Eq.~\eqref{eq:gps1l-depth} is \emph{exact at $r=0$}; for $r>0$ it is a leading-order \emph{upper bound}, since it charges the brickwork of each scale separately whereas overlapping layers from different scales only reduce the total. Its exact $r>0$ value does not enter the complexity analysis, which depends only on the $r=0$ term, all $r>0$ contributions being doubly-exponentially subdominant.

\section{Complete GPF Compilation}
\label{sec:full}

This section collects the paper's core conclusions, obtained by combining the two module theorems along the GPF recursion. We first assemble the blocks, then prove the main result---restricted-topology routing preserves the complexity in closed form (Theorems~\ref{thm:route-depth}--\ref{thm:route-count})---with topology-independence following as a corollary (Corollary~\ref{cor:topology}). We then contrast this with general-purpose routers on the native two-row topology, and give the ancilla space--time extension.

\subsection{Complete GPF as Recursive Assembly of Unit Blocks}
\label{sec:assembly}

The synthesis construction of the first half of this paper fully determines the complete GPF circuit as a fixed sequence of GPS/GPS$^*$ blocks (emitted in a deterministic order); compilation need only route each block using Section~\ref{sec:tworow} or \ref{sec:onerow} and concatenate them in synthesis order. The recursive structure is the one summarized in Fig.~\ref{fig:recursion-tree}: with $n_H=\lceil n/2\rceil$ and $n_L=\lfloor n/2\rfloor$, $\GPF(n)$ consists of (i)~top-level GPS-2L mixings on the $2\times n_H$ grid (Section~\ref{sec:tworow}), (ii)~a high-half $\GPF(n_H)$ on a single row with GPS-1L, and (iii)~a low-half $\GPF^*(n_L)$ on a single row with GPS$^*$-1L (Section~\ref{sec:onerow}); two-row geometry exists only at the top level.

The number of top-level GPS-2L blocks is given by a structural fact established in the synthesis half (Section~\ref{sec:depth}): the GPS reference count equals the $R_z$ layer depth of the low-half GPF$^*$ circuit, i.e.\ the number of Walsh-mode groups, denoted $D^{\GPF^*}_z(n_L)$ and satisfying the phase-layer recurrence of Eq.~\eqref{eq:recurrences} (its first line). For example, $D^{\GPF^*}_z(2),\ldots,D^{\GPF^*}_z(6)=2,4,6,12,20$, which are exactly the top-level GPS-2L block counts for $n=4,\ldots,12$. Each top-level block has $2^{n_H}$ logical CNOT layers (one offset ring, Section~\ref{sec:tworow}).

Each block's routing returns to offset $s_0$: the top-level offset ring starts from $s_0$ and closes back to $s_0$ (Lemma~\ref{lem:offset-ring}), and single-row blocks' brickwork likewise returns to a zero-offset arrangement at block end. A block ends in the canonical arrangement up to an $O(1)$ boundary correction $b$ (Section~\ref{sec:jump-cost}; for low-jump, $b=0$ and the block self-closes exactly to the canonical arrangement for $n_H\le 6$); any residual permutation of the $n_H$ lanes is restored to canonical by at most $n_H$ brick layers, i.e.\ $O(n_H)$ additional SWAP depth per block. The full circuit is then the head-to-tail concatenation of each block's routed circuit, with all two-qubit gates on nearest-neighbor edges. The complete routed circuit for $n=8$---every block of $\GPF(8)$ routed and concatenated, rendered as a 2.5D depth chart---is in Appendix~\ref{app:circuit}, Fig.~\ref{fig:gpf-routed-n8}: $278$ CNOTs plus $350$ SWAPs, all nearest-neighbor, realizing the original $\GPF(8)$. Summed over the $D^{\GPF^*}_z(n_L)=\Theta(2^{n/2})$ top-level blocks, these per-block boundary and restoration costs contribute at most $O(n_H\,2^{n/2})=O(2^{n/2})$ to the total routing depth---doubly-exponentially subdominant to the leading $\Theta(2^n/n)$ term, so they leave the leading constant unchanged.

\subsection{Routing-Cost Complexity}
\label{sec:complexity}

Summing single-block costs along the recursion tree, each block's routing cost is proportional to its own logical cost:
\begin{itemize}
    \item \textbf{Two-row top-level GPS-2L} (width $n_H$): routing depth $3\mu_{\mathrm{jump}}2^{n_H}$, gate count $3(n_H-1)\mu_{\mathrm{jump}}2^{n_H}$ (Theorem~\ref{thm:gps2l});
    \item \textbf{Single-row GPS-1L/GPS$^*$-1L} (width $w_r=n/2^{r+1}$, single-row recursion depth $r$): routing depth $6\mu_{\mathrm{jump}}2^r 2^{w_r}$, gate count $3(2w_r-1)\mu_{\mathrm{jump}}2^r 2^{w_r}$ (Theorem~\ref{thm:gps1l}; the $2^r$ factor is the bystander-induced brickwork deepening, from $\delta_{r+3}+\delta_{r+2}=2^{r+2}$; at $r=0$ exactly twice the two-row cost).
\end{itemize}

Only the top-level two-row term is needed in exact form: it dominates (proof below), so the single-row terms---exact at their own $r=0$ and leading-order upper bounds for $r>0$---enter the result only through the doubly-exponentially subdominant correction. The leading constants below are therefore exact in $C$, and exact in $\mu_{\mathrm{jump}}$ for $n_H\le 6$ (upper bounds in the large-$n$ regime where $\mu_{\mathrm{jump}}\approx 1.06$).

We adopt the complexity notation of the synthesis half: the constant $C\approx 3.40147$ (Eq.~\eqref{eq:const-C}) and its partial products $P_m\to C$ (Appendix~\ref{app:constants}); the routing constant additionally carries $\mu_{\mathrm{jump}}$ ($\approx 1.06$ for the low-jump instance, Lemma~\ref{lem:minjump}).

\begin{theorem}[Routing depth complexity of complete GPF]
\label{thm:route-depth}
The routing-added CNOT depth of complete $\GPF(n)$ satisfies the recurrences
\begin{align}
    D^{\GPF\text{-}2\mathrm{L}}_{\mathrm{route}}(n) ={}
    &\underbrace{3\mu_{\mathrm{jump}}\,2^{n_H}D^{\GPF^*}_z(n_L)}_{\text{top-level GPS-2L}}
    \notag\\
    &+ D^{\GPF\text{-}1\mathrm{L}}_{\mathrm{route}}(n_H,0)
    \notag\\
    &+ D^{\GPF^*\text{-}1\mathrm{L}}_{\mathrm{route}}(n_L,0),
    \label{eq:route-depth-top} \\
    D^{\GPF\text{-}1\mathrm{L}}_{\mathrm{route}}(n,r) ={}
    &6\mu_{\mathrm{jump}}\,2^r 2^{n_H}D^{\GPF^*}_z(n_L)
    \notag\\
    &+ D^{\GPF\text{-}1\mathrm{L}}_{\mathrm{route}}(n_H,r{+}1)
    \notag\\
    &+ D^{\GPF^*\text{-}1\mathrm{L}}_{\mathrm{route}}(n_L,r{+}1),
    \label{eq:route-depth-1l} \\
    D^{\GPF^*\text{-}1\mathrm{L}}_{\mathrm{route}}(n,r) ={}
    &6\mu_{\mathrm{jump}}\,2^r 2^{n_H}D^{\GPF^*}_z(n_L)
    \notag\\
    &+ D^{\GPF^*\text{-}1\mathrm{L}}_{\mathrm{route}}(n_H,r{+}1)
    \notag\\
    &+ D^{\GPF^*\text{-}1\mathrm{L}}_{\mathrm{route}}(n_L,r{+}1)
    \label{eq:route-depth-star1l}
\end{align}
($n\le 2$ yields $0$). The two single-row recurrences share the same leading term and differ only in the high-half module (GPF vs.\ GPF$^*$), so their distinction affects only the doubly-exponentially subdominant correction, not the leading constant. The normalized depth $D^{\GPF\text{-}2\mathrm{L}}_{\mathrm{route}}(n)/(2^n/n)$ oscillates in sawtooth fashion, with minima at $n=2^\lambda$ and maxima at $n=2^{\lambda+1}-1$ within each dyadic interval $[2^\lambda,2^{\lambda+1})$, converging respectively to
\begin{equation}
\begin{aligned}
    \frac{D^{\GPF\text{-}2\mathrm{L}}_{\mathrm{route}}(2^\lambda)}{2^n/n}
    &\to 3\mu_{\mathrm{jump}}\,C\approx 10.82,\\
    \frac{D^{\GPF\text{-}2\mathrm{L}}_{\mathrm{route}}(2^\lambda{-}1)}{2^n/n}
    &\to 3\mu_{\mathrm{jump}}\,2(C{-}1)\approx 15.27.
\end{aligned}
\end{equation}
\end{theorem}

\begin{proof}[Proof sketch]
The top-level term $3\mu_{\mathrm{jump}}2^{n_H}D^{\GPF^*}_z(n_L)$ dominates: the two single-row recursions act on roughly $n/2$ qubits each, with magnitude $\Theta(2^{n/2})$, doubly-exponentially smaller than the top-level's $\Theta(2^n/n)$ (relative correction $\sim 2^{-n/2}$). On the balanced subsequence $n=2^\lambda$ with $n_H=n_L=2^{\lambda-1}$, the phase recurrence collapses to a scalar form, and telescoping gives the top-level term divided by $2^n/n$ approaching $3\mu_{\mathrm{jump}}C$; the upper envelope from the maximally unbalanced $n=2^\lambda-1$ shifts the base constant from $C$ to $2(C-1)$ by the same telescoping argument. Full algebra is in Appendix~\ref{app:route-constants}. Evaluating the recurrence at $n=8,16,32,64$ gives normalized depth $12.0,11.1,10.82,10.82$, monotonically approaching $3\mu_{\mathrm{jump}}C$ (Fig.~\ref{fig:route-depth}).
\end{proof}

\begin{theorem}[Routing gate-count complexity of complete GPF]
\label{thm:route-count}
The routing-added CNOT gate count $S^{\GPF\text{-}2\mathrm{L}}_{\mathrm{route}}(n)$ satisfies an isomorphic recurrence (top-level: $3(n_H-1)\mu_{\mathrm{jump}}2^{n_H}$ per block; single-row: $3(2w_r-1)\mu_{\mathrm{jump}}2^r 2^{w_r}$ per block). The normalized gate count $S^{\GPF\text{-}2\mathrm{L}}_{\mathrm{route}}(n)/2^n$ oscillates in sawtooth fashion with subsequence limits
\begin{equation}
\begin{aligned}
    \frac{S^{\GPF\text{-}2\mathrm{L}}_{\mathrm{route}}(2^\lambda)}{2^n}
    &\to\tfrac{3}{2}\mu_{\mathrm{jump}}\,C\approx 5.41,\\
    \frac{S^{\GPF\text{-}2\mathrm{L}}_{\mathrm{route}}(2^\lambda{-}1)}{2^n}
    &\to\tfrac{3}{2}\mu_{\mathrm{jump}}\,2(C{-}1)\approx 7.64.
\end{aligned}
\end{equation}
Gate count exceeds depth by a width factor of approximately $n_H$ (each depth layer has roughly $n_H-1$ parallel SWAPs); under normalization this factor is absorbed into $2^n$, halving the constant to $\tfrac{3}{2}\mu_{\mathrm{jump}}$. The proof is isomorphic to Theorem~\ref{thm:route-depth} (Fig.~\ref{fig:route-count}).
\end{theorem}

\begin{figure}[!t]
    \centering
    \includegraphics[width=\columnwidth]{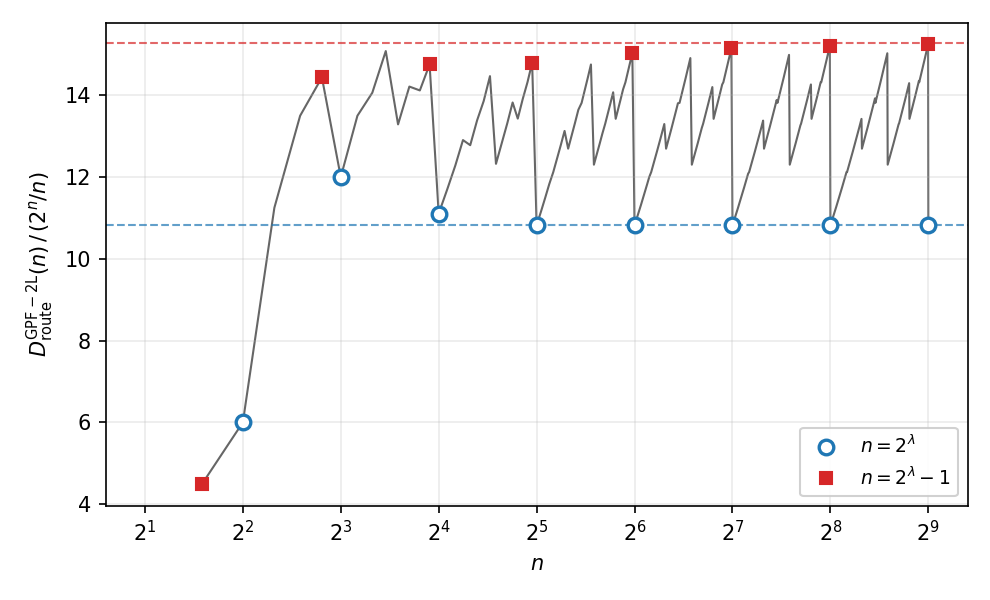}
    \caption{Normalized routing-depth complexity of complete GPF. $D^{\GPF\text{-}2\mathrm{L}}_{\mathrm{route}}/(2^n/n)$ oscillates in sawtooth fashion: hollow circles ($n=2^\lambda$) form the lower envelope, converging to $3\mu_{\mathrm{jump}}C\approx 10.82$; filled squares ($n=2^\lambda-1$) form the upper envelope, converging to $3\mu_{\mathrm{jump}}\cdot 2(C-1)\approx 15.27$.}
    \label{fig:route-depth}
\end{figure}

\begin{figure}[!t]
    \centering
    \includegraphics[width=\columnwidth]{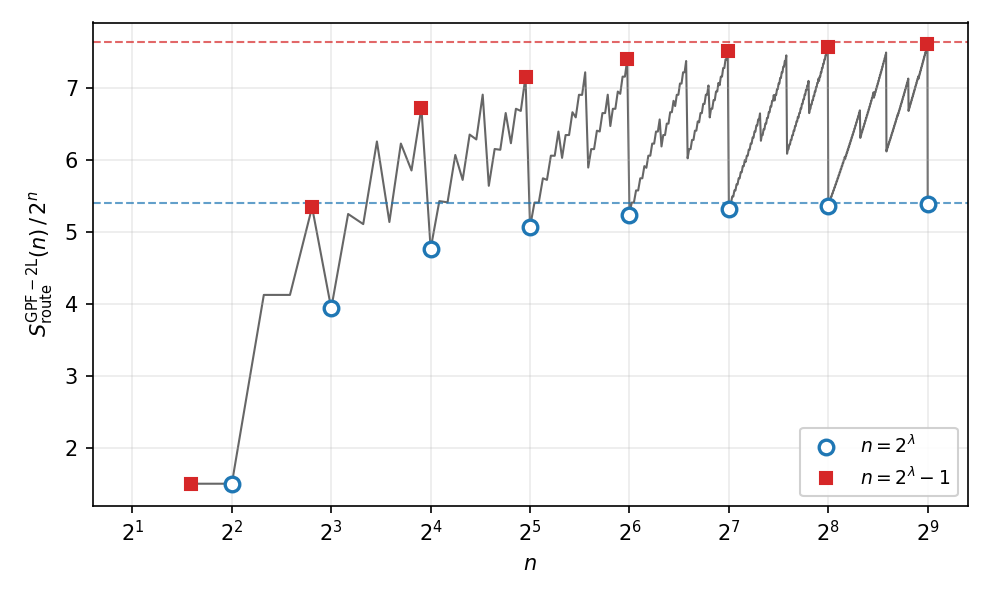}
    \caption{Normalized routing gate-count complexity of complete GPF. $S^{\GPF\text{-}2\mathrm{L}}_{\mathrm{route}}/2^n$ exhibits isomorphic sawtooth behavior: lower envelope converges to $\tfrac{3}{2}\mu_{\mathrm{jump}}C\approx 5.41$; upper envelope to $\tfrac{3}{2}\mu_{\mathrm{jump}}\cdot 2(C-1)\approx 7.64$.}
    \label{fig:route-count}
\end{figure}

As an end-to-end consistency check, the single-block formulas, the assembly recurrence, and an explicitly constructed circuit agree at the smallest fully-recursive instance, $n=8$. The block-count recurrence~\eqref{eq:recurrences} reproduces the top-level GPS-2L counts $D^{\GPF^*}_z(2),\ldots,D^{\GPF^*}_z(6)=2,4,6,12,20$; the recurrence~\eqref{eq:route-depth-top}--\eqref{eq:route-depth-star1l} evaluated at $n=8$ gives normalized depth $12.0$, on the predicted sawtooth (Fig.~\ref{fig:route-depth}); and the fully assembled, concatenated $\GPF(8)$ routed circuit (Section~\ref{sec:assembly}, Appendix~\ref{app:circuit}) is the explicitly constructed circuit realizing the original. The complexity statement therefore rests on a chain that closes end-to-end at $n=8$, not on the asymptotics alone.

This yields the paper's main result: the routing overhead on restricted topology is of the same asymptotic order as synthesis---depth $\Theta(2^n/n)$, gate count $\Theta(2^n)$---with leading constants given in closed form by $C$ and $\mu_{\mathrm{jump}}$ ($\mu_{\mathrm{jump}}\approx 1.06$ for the low-jump instance). Compiling GPF onto a nearest-neighbor topology does not change its complexity class; it introduces only a definite constant factor. This preservation is also topology-independent, and on any topology containing a linear chain the compiled depth is asymptotically optimal, as we now record.

\begin{corollary}[Topology-independent preservation and optimality]
\label{cor:topology}
Folding the top-level block onto a single row as well (Theorem~\ref{thm:gps1l} at $r=0$) compiles the complete $\GPF(n)$ on a one-dimensional line at routing depth $\Theta(2^n/n)$, at twice the two-row leading constant. A line is a spanning subgraph of every standard nearest-neighbor architecture---a $1$-D chain, a $2$-D grid, a ladder, or heavy-hex---so the same routing applies on each. Moreover the $\Omega(2^n/n)$ depth lower bound for a generic $n$-qubit diagonal unitary is a counting bound on the at most $\lfloor n/2\rfloor$ two-qubit gates a single layer admits, independent of connectivity \cite{Sun2023Depth}, so it holds on each of these topologies and the construction matches it. The compiled depth is therefore asymptotically optimal on any topology containing a linear chain, with the two-row grid (Theorem~\ref{thm:gps2l}) the cost-optimized choice whose row dimension the ancilla fan-out scales.
\end{corollary}

\subsection{Comparison with General-Purpose Routers}
\label{sec:comparison}

This comparison is the empirical face of the preservation result: on the native two-row layout our normalized routing overhead stays bounded and closed-form, while general-purpose heuristic routers must search and become infeasible at these sizes. To make this concrete, we place our method and the routers on the same $2\times n_H$ two-row topology. Three representative routers are used: qiskit's LightSABRE \cite{Li2019SABRE}, tket's LexiRoute \cite{Sivarajah2021tket}, and BQSKit's Sabre \cite{Younis2021BQSKit}. Each is given the same low-jump GPF CNOT skeleton (Section~\ref{sec:minjump}) and performs only layout and routing without additional gate cancellation, matching the construction-only, no-local-optimization setting of our method. The unified metric is the post-routing total two-qubit depth: SWAPs are expanded to $3$ CNOTs and tket's BRIDGE remote CNOT to $4$ CNOTs, with no cancellation, and our method uses the conservative upper bound from Section~\ref{sec:complexity}. Scans cover $n=3,\ldots,16$; search-based routers become infeasible beyond $n\approx 16$.


\begin{figure}[!t]
    \centering
    \includegraphics[width=\columnwidth]{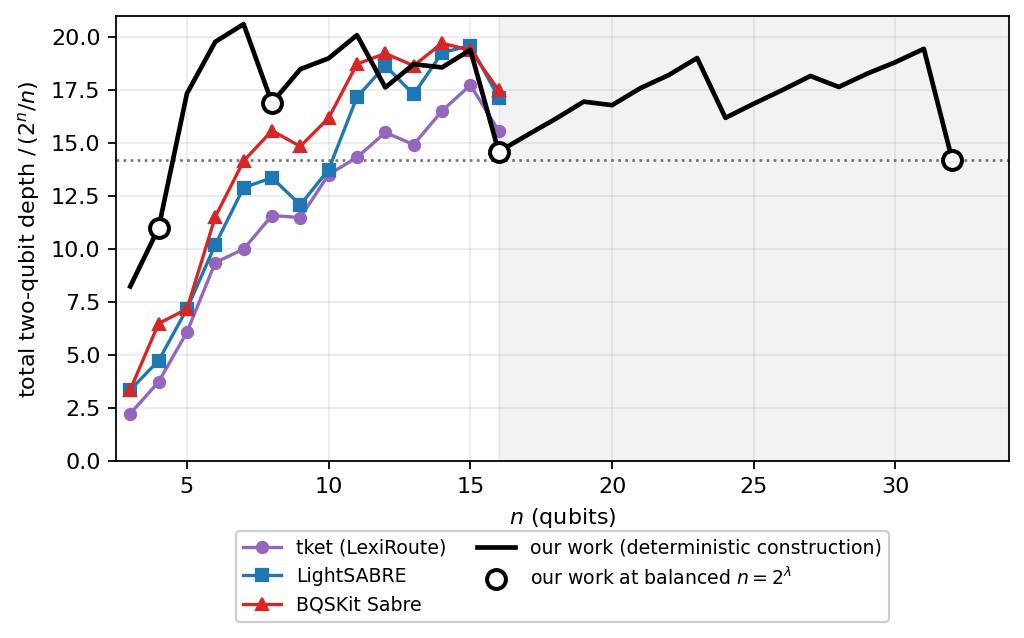}
    \caption{Normalized total two-qubit depth on the $2\times n_H$ two-row topology (divided by $2^n/n$). Our deterministic construction's normalized depth is bounded (converging to a constant) and needs no search, so it extends across the shaded region where search-based routers become infeasible.}
    \label{fig:comparison}
\end{figure}

Fig.~\ref{fig:comparison} normalizes the total depth by $2^n/n$. We flag one asymmetry here, revisited in caveat~(i) below: our curve is the analytic upper bound of Section~\ref{sec:complexity} while the router curves are measured, so the comparison is conservative, our true cost lying below the bound. The primary finding is structural rather than a single-point win: our normalized total depth converges to the closed-form constant $3\mu_{\mathrm{jump}}C+C\approx 14.2$ (Section~\ref{sec:complexity}) and is available without any search, so it extends to scales where search-based routers are infeasible ($n\gtrsim 16$, shaded region). The curves cross at the balanced size $n=16$, where our (upper-bound) cost first falls below all three routers; we stress that this is one point on a trend, not the main claim, and that tket's LexiRoute is the strongest baseline at the smaller sizes $n=8,10,12,14$.

Three caveats keep this comparison honest. \emph{(i) Bound vs.\ measured.} Our curve in Fig.~\ref{fig:comparison} is the conservative upper bound of Section~\ref{sec:complexity}---every layer at full cost, including the $2^s$ bystander factor, with no logical/routing overlap---whereas the routers are measured; the assembled circuit is markedly sparser (the assembled $\GPF(8)$ uses $278$ CNOTs and $350$ SWAPs, Appendix~\ref{app:circuit}), so the true ``Ours'' values lie below the bound and the crossover is conservative. \emph{(ii) Topology.} The construction itself is topology-independent (Corollary~\ref{cor:topology}); the $2\times n_H$ two-row grid is the cost-optimized layout \emph{induced by the GPF recursion}, not a near-square coupling map tuned for generic circuits. What this comparison measures is the \emph{constant} on that geometry. On a more connected near-square grid the general routers exploit the extra edges and lead on the constant, so the bounded-overhead advantage is specific to the two-row geometry GPF targets---``native'' means native to the method, not the device. \emph{(iii) Metric.} ``Total two-qubit depth'' is a connectivity-overhead proxy; our routing is SWAP-heavy ($n=8$ adds $350$ routing SWAPs---$1050$ CNOT-equivalents---against $278$ logical CNOTs).

\subsection{Beyond Ancilla-Free: A Depth--Space Tradeoff}
\label{sec:fanout}

GPF's Walsh phases are additive and order-independent (a byproduct of the block-wise construction of Section~\ref{sec:assembly}), so the $\Theta(2^n/n)$ phase layers can be split into $k$ portions and run on $k$ copies of the working register in parallel. The gain is not confined to the phase count: in a copy that keeps only its $1/k$ portion of phases, every CNOT-skeleton segment bracketing an omitted phase has no intervening phase gate and cancels pairwise---a GPS/GPS$^*$ block stripped of all phases is the identity, since the closed Gray cycle returns every parity check to its initial value. Each copy therefore collapses to roughly $1/k$ of the full circuit plus a small boundary residual, so both its CNOT and its phase depth fall to $\approx D^{\GPF\text{-}2\mathrm{L}}_{\mathrm{route}}(n)/k$. This depth--space trade-off is a third realization of the construction, spending ancilla rows to buy depth, and concretizes known depth--space trade-offs for diagonal operators with ancillas \cite{Sun2023Depth,Jiang2020SpaceDepth,Zylberman2024Diagonal}.

\begin{figure}[!t]
    \centering
    \includegraphics[width=0.82\columnwidth]{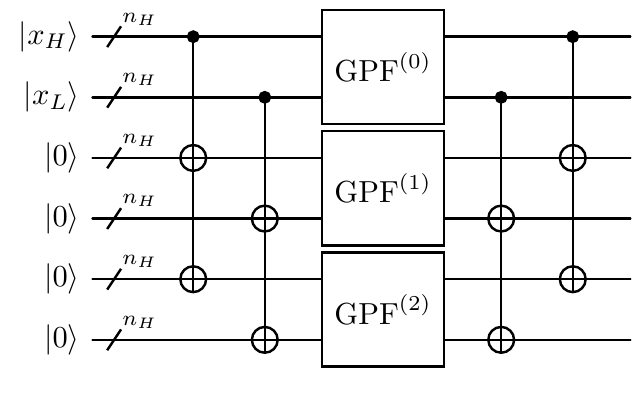}
    \vspace{-0.6em}
    \caption{Fan-out replication of GPF phase layers ($k=3$): the working register is copied to $k-1$ ancilla copies, each running its $1/k$ share of the phase layers in parallel before un-fan-out accumulates them back. This is the space--time extension---$k$-fold space cuts depth to $\approx 1/k$.}
    \label{fig:fanout}
\end{figure}

The mechanism has three steps: fan-out replicates the working register into $k$ copies via lateral CNOTs; each copy $j$ then executes $\GPF^{(j)}$---the shared CNOT skeleton carrying only its $j$-th portion of phases, the rest cancels as above---with all $k$ copies in parallel; and un-fan-out applies reverse CNOTs to accumulate every copy's phases back onto the working register and reset the ancillas to $\ket{0}$.

The boundary residual is bounded as follows. The only skeleton layers surviving the pairwise cancellation in a copy are the partial blocks at the two ends of its retained phase range; each such partial block simplifies to at most $n_H$ layers (slightly more at deeper recursion, from bystander distances), and at most one arises per level of the $O(\log n)$-deep recursion tree. Summing over levels gives $O(n_H\log n)=O(n\log n)$ surviving layers per copy, to which fan-out and un-fan-out add only a lower-order replication cost; the constant-tracking is a routine induction we omit.

Using $k$ copies on $kn$ physical qubits, each copy runs $\approx 1/k$ of the circuit, so
\begin{equation}
    D^{\mathrm{ancilla}}_{\mathrm{route}}(kn) \approx \frac{D^{\GPF\text{-}2\mathrm{L}}_{\mathrm{route}}(n)}{k}+O(n\log n),
\end{equation}
where $D^{\GPF\text{-}2\mathrm{L}}_{\mathrm{route}}(n)=\Theta(2^n/n)$ is the complete GPF routing depth of Section~\ref{sec:complexity}: $k$-fold space amortizes the exponential leading term by $k$ at an additive $O(n\log n)$ cost, negligible against $2^n/n$. This stays asymptotically optimal for the enlarged register: with $a$ ancillary qubits the depth of an $n$-qubit diagonal unitary is bounded below by $\Omega\!\big(2^n/(n+a)\big)$ \cite{Sun2023Depth}, which for the $2k$-row layout's $a=(k{-}1)n$ ancillas reads $\Omega(2^n/(kn))$; the routed depth above meets this bound up to the additive $O(n\log n)$ term, so the trade-off stays on the optimal frontier even once routing overhead is counted. Two caveats. First, replication needs cross-copy connectivity that a single $2\times n_H$ working grid lacks; each copy occupies its own two-row grid, so the $k$ copies tile an extended grid of $2k$ rows of $n_H$ (the working register's two rows plus $2(k{-}1)$ ancilla rows; Fig.~\ref{fig:fanout}), on which the lateral replication CNOTs become inter-row edges, so the natural setting is this larger grid rather than all-to-all connectivity. Second, the scheme abandons the ancilla-free premise that motivates GPF. The exact extended-grid dimensions follow from the same recursion; we do not track their constants here.

\section{Conclusion}
\label{sec:discussion}

This paper treats the circuit realization of diagonal unitaries end to end. On the synthesis side, an ancilla-free Gray-Path Framework realizes any $n$-qubit diagonal unitary in $R_z$ and CNOT depth $O(2^n/n)$, governed by a single explicit constant $C\approx 3.40147$. On the compilation side we establish three results. First, mapping the construction onto a nearest-neighbor two-row grid preserves this asymptotic optimality with closed-form constants: because the interaction structure is fixed in advance, routing is solved by a closed-form schedule rather than heuristic search, so restricted topology costs only a definite constant factor. Second, folding the construction onto a single row extends this preservation to be topology-independent---asymptotically optimal on any architecture containing a linear chain---at twice the two-row constant. Third, replicating the phase layers across additional rows yields a space--time tradeoff that cuts depth by a factor $k$ using $k$ working copies while remaining asymptotically optimal for the enlarged register. Together these place diagonal-unitary compilation on restricted hardware on the same asymptotic footing as all-to-all synthesis, with the leading constants available in closed form.

A natural open question is whether the asymptotic depth optimality established for diagonal unitaries extends to larger classes of structured operators under nearest-neighbor constraints.

\bibliographystyle{IEEEtran}
\bibliography{article}

\begin{thebibliography}{10}
\providecommand{\url}[1]{#1}
\csname url@samestyle\endcsname
\providecommand{\newblock}{\relax}
\providecommand{\bibinfo}[2]{#2}
\providecommand{\BIBentrySTDinterwordspacing}{\spaceskip=0pt\relax}
\providecommand{\BIBentryALTinterwordstretchfactor}{4}
\providecommand{\BIBentryALTinterwordspacing}{\spaceskip=\fontdimen2\font plus
\BIBentryALTinterwordstretchfactor\fontdimen3\font minus
  \fontdimen4\font\relax}
\providecommand{\BIBforeignlanguage}[2]{{%
\expandafter\ifx\csname l@#1\endcsname\relax
\typeout{** WARNING: IEEEtran.bst: No hyphenation pattern has been}%
\typeout{** loaded for the language `#1'. Using the pattern for}%
\typeout{** the default language instead.}%
\else
\language=\csname l@#1\endcsname
\fi
#2}}
\providecommand{\BIBdecl}{\relax}
\BIBdecl

\bibitem{Farhi2014QAOA}
E.~Farhi, J.~Goldstone, and S.~Gutmann, ``A quantum approximate optimization
  algorithm,'' 2014.

\bibitem{Bremner2011IQP}
M.~J. Bremner, R.~Jozsa, and D.~J. Shepherd, ``Classical simulation of
  commuting quantum computations implies collapse of the polynomial
  hierarchy,'' \emph{Proceedings of the Royal Society A: Mathematical, Physical
  and Engineering Sciences}, vol. 467, no. 2126, pp. 459--472, 2011.

\bibitem{Welch2014Diagonal}
J.~Welch, D.~Greenbaum, S.~Mostame, and A.~Aspuru-Guzik, ``Efficient quantum
  circuits for diagonal unitaries without ancillas,'' \emph{New Journal of
  Physics}, vol.~16, no.~3, p. 033040, 2014.

\bibitem{Sun2023Depth}
X.~Sun, G.~Tian, S.~Yang, P.~Yuan, and S.~Zhang, ``Asymptotically optimal
  circuit depth for quantum state preparation and general unitary synthesis,''
  \emph{IEEE Transactions on Computer-Aided Design of Integrated Circuits and
  Systems}, vol.~42, no.~10, pp. 3301--3314, 2023.

\bibitem{Zhang2024DiagonalDepth}
S.~Zhang, K.~Huang, and L.~Li, ``Depth-optimized quantum circuit synthesis for
  diagonal unitary operators with asymptotically optimal gate count,''
  \emph{Physical Review A}, vol. 109, no.~4, p. 042601, 2024.

\bibitem{Yan2024Overview}
G.~Yan, W.~Wu, Y.~Chen, K.~Pan, X.~Lu, Z.~Zhou, Y.~Wang, R.~Wang, and J.~Yan,
  ``Quantum circuit synthesis and compilation optimization: Overview and
  prospects,'' 2024.

\bibitem{Li2019SABRE}
G.~Li, Y.~Ding, and Y.~Xie, ``Tackling the qubit mapping problem for {NISQ}-era
  quantum devices,'' in \emph{Proceedings of the 24th International Conference
  on Architectural Support for Programming Languages and Operating Systems
  (ASPLOS '19)}, 2019, pp. 1001--1014.

\bibitem{Shende2006Synthesis}
V.~V. Shende, S.~S. Bullock, and I.~L. Markov, ``Synthesis of quantum-logic
  circuits,'' \emph{IEEE Transactions on Computer-Aided Design of Integrated
  Circuits and Systems}, vol.~25, no.~6, pp. 1000--1010, 2006.

\bibitem{Bullock2004Diagonal}
S.~S. Bullock and I.~L. Markov, ``Asymptotically optimal circuits for arbitrary
  n-qubit diagonal computations,'' \emph{Quantum Information \& Computation},
  vol.~4, no.~1, pp. 27--47, 2004.

\bibitem{Amy2018CNOTPhase}
M.~Amy, P.~Azimzadeh, and M.~Mosca, ``On the controlled-{NOT} complexity of
  controlled-{NOT}-phase circuits,'' \emph{Quantum Science and Technology},
  vol.~4, no.~1, p. 015002, 2018.

\bibitem{Vandaele2022PhasePolynomial}
V.~Vandaele, S.~Martiel, and T.~Goubault~de Brugi{\`e}re, ``Phase polynomials
  synthesis algorithms for {NISQ} architectures and beyond,'' \emph{Quantum
  Science and Technology}, vol.~7, no.~4, p. 045027, 2022.

\bibitem{Kivlichan2018SwapNetwork}
I.~D. Kivlichan, J.~McClean, N.~Wiebe, C.~Gidney, A.~Aspuru-Guzik, G.~K.-L.
  Chan, and R.~Babbush, ``Quantum simulation of electronic structure with
  linear depth and connectivity,'' \emph{Physical Review Letters}, vol. 120,
  no.~11, p. 110501, 2018.

\bibitem{OGorman2019SwapNetworks}
B.~O'Gorman, W.~J. Huggins, E.~G. Rieffel, and K.~B. Whaley, ``Generalized swap
  networks for near-term quantum computing,'' 2019.

\bibitem{Kotil2023QAOARouting}
A.~Kotil, F.~{\v{S}}imkovic, and M.~Leib, ``Improved qubit routing for {QAOA}
  circuits,'' 2023.

\bibitem{Zhu2024Coqa}
Y.~Zhu, Y.~Zhou, J.~Cheng, Y.~Jin, B.~Li, S.~Niu, and Z.~Liang, ``Compiler
  optimizations for {QAOA},'' in \emph{Proceedings of the 43rd IEEE/ACM
  International Conference on Computer-Aided Design (ICCAD '24)}, 2024, pp.
  1--7.

\bibitem{Hashim2022FermionicSwap}
A.~Hashim, R.~Rines, V.~Omole, R.~K. Naik, J.~M. Kreikebaum, D.~I. Santiago,
  F.~T. Chong, I.~Siddiqi, and P.~Gokhale, ``Optimized {SWAP} networks with
  equivalent circuit averaging for {QAOA},'' \emph{Physical Review Research},
  vol.~4, p. 033028, 2022.

\bibitem{Montanez2025ParityTwine}
J.~A. Monta{\~n}ez-Barrera, Y.~Ji, M.~R. von Spakovsky, D.~E. Bernal~Neira, and
  K.~Michielsen, ``Optimizing {QAOA} circuit transpilation with parity twine
  and {SWAP} network encodings,'' 2025.

\bibitem{Jin2021StructuredQAOA}
Y.~Jin, J.~Luo, L.~Fong, Y.~Chen, A.~B. Hayes, C.~Zhang, F.~Hua, and E.~Z.
  Zhang, ``A structured method for compilation of {QAOA} circuits in quantum
  computing,'' 2021.

\bibitem{Kissinger2020SteinerGauss}
A.~Kissinger and A.~Meijer-van~de Griend, ``{CNOT} circuit extraction for
  topologically-constrained quantum memories,'' \emph{Quantum Information \&
  Computation}, vol.~20, no. 7\&8, pp. 581--596, 2020.

\bibitem{Nash2020NISQ}
B.~Nash, V.~Gheorghiu, and M.~Mosca, ``Quantum circuit optimizations for {NISQ}
  architectures,'' \emph{Quantum Science and Technology}, vol.~5, no.~2, p.
  025010, 2020.

\bibitem{Griend2023PermRowCol}
A.~Meijer-van~de Griend and S.~M. Li, ``Dynamic qubit routing with {CNOT}
  circuit synthesis for quantum compilation,'' in \emph{Proceedings of the 19th
  International Conference on Quantum Physics and Logic (QPL 2022)}, ser.
  Electronic Proceedings in Theoretical Computer Science, vol. 394, 2023.

\bibitem{Brugiere2022Depth}
T.~Goubault~de Brugi{\`e}re, M.~Baboulin, B.~Valiron, S.~Martiel, and
  C.~Allouche, ``Reducing the depth of linear reversible quantum circuits,''
  \emph{IEEE Transactions on Quantum Engineering}, vol.~2, pp. 1--22, 2021.

\bibitem{Yuan2024Connectivity}
P.~Yuan, J.~Allcock, and S.~Zhang, ``Does qubit connectivity impact quantum
  circuit complexity?'' \emph{IEEE Transactions on Computer-Aided Design of
  Integrated Circuits and Systems}, vol.~43, no.~2, pp. 520--533, 2024.

\bibitem{Wu2023CNOTLimited}
B.~Wu, X.~He, S.~Yang, L.~Shou, G.~Tian, J.~Zhang, and X.~Sun, ``Optimization
  of {CNOT} circuits on limited-connectivity architecture,'' \emph{Physical
  Review Research}, vol.~5, no.~1, p. 013065, 2023.

\bibitem{Jiang2020SpaceDepth}
J.~Jiang, X.~Sun, S.-H. Teng, B.~Wu, K.~Wu, and J.~Zhang, ``Optimal space-depth
  trade-off of {CNOT} circuits in quantum logic synthesis,'' in
  \emph{Proceedings of the 2020 ACM-SIAM Symposium on Discrete Algorithms
  (SODA)}, 2020, pp. 213--229.

\bibitem{Huang2026ZAP}
C.-H. Huang, X.~Zhao, H.-Z. Xu, W.~Zhuang, M.-J. Hu, D.~E. Liu, and J.~Wang,
  ``{ZAP}: Zoned architecture and performant compiler for field programmable
  atom array,'' \emph{IEEE Transactions on Quantum Engineering}, 2026, early
  access.

\bibitem{Zylberman2024Diagonal}
J.~Zylberman, U.~Nzongani, A.~Simonetto, and F.~Debbasch, ``Efficient quantum
  circuits for non-unitary and unitary diagonal operators with
  space-time-accuracy trade-offs,'' \emph{ACM Transactions on Quantum
  Computing}, vol.~6, no.~2, pp. 1--43, 2025.

\bibitem{Patel2008LinearReversible}
K.~N. Patel, I.~L. Markov, and J.~P. Hayes, ``Optimal synthesis of linear
  reversible circuits,'' \emph{Quantum Information \& Computation}, vol.~8, no.
  3--4, pp. 282--294, 2008.

\bibitem{Mottonen2004Multiqubit}
M.~M{\"o}tt{\"o}nen, J.~J. Vartiainen, V.~Bergholm, and M.~M. Salomaa,
  ``Quantum circuits for general multiqubit gates,'' \emph{Physical Review
  Letters}, vol.~93, no.~13, p. 130502, 2004.

\bibitem{Litinski2019SurfaceCodes}
D.~Litinski, ``A game of surface codes: Large-scale quantum computing with
  lattice surgery,'' \emph{Quantum}, vol.~3, p. 128, 2019.

\bibitem{Sivarajah2021tket}
S.~Sivarajah, S.~Dilkes, A.~Cowtan, W.~Simmons, A.~Edgington, and R.~Duncan,
  ``t\textbar ket\ensuremath{\rangle}: a retargetable compiler for {NISQ}
  devices,'' \emph{Quantum Science and Technology}, vol.~6, no.~1, p. 014003,
  2021.

\bibitem{Younis2021BQSKit}
\BIBentryALTinterwordspacing
E.~Younis, C.~Iancu, W.~Lavrijsen, M.~Davis, and E.~Smith, ``Berkeley quantum
  synthesis toolkit ({BQSKit}),'' Lawrence Berkeley National Laboratory, 2021.
  [Online]. Available: \url{https://github.com/BQSKit/bqskit}
\BIBentrySTDinterwordspacing

\bibitem{SacHimelfarb2025SkewTolerant}
G.~Sac~Himelfarb and M.~Schwartz, ``Improved constructions of skew-tolerant
  gray codes,'' \emph{IEEE Transactions on Information Theory}, vol.~71,
  no.~10, pp. 8017--8028, 2025.

\end{thebibliography}

\clearpage
\onecolumn
\appendices

\renewcommand{\theequation}{\thesection\arabic{equation}}

\begin{sidewaysfigure}[!p]
	\refstepcounter{section}
	\label{app:gpf8} 
	
	\addcontentsline{toc}{section}{\protect\numberline{\thesection}Full Eight-Qubit GPF Circuit}
	
	\centering
	{\LARGE\bfseries Appendix \thesection\quad Full Eight-Qubit GPF Circuit \par}
	\vspace{5em} 
	
	\centering
	\includegraphics[width=0.9\textheight]{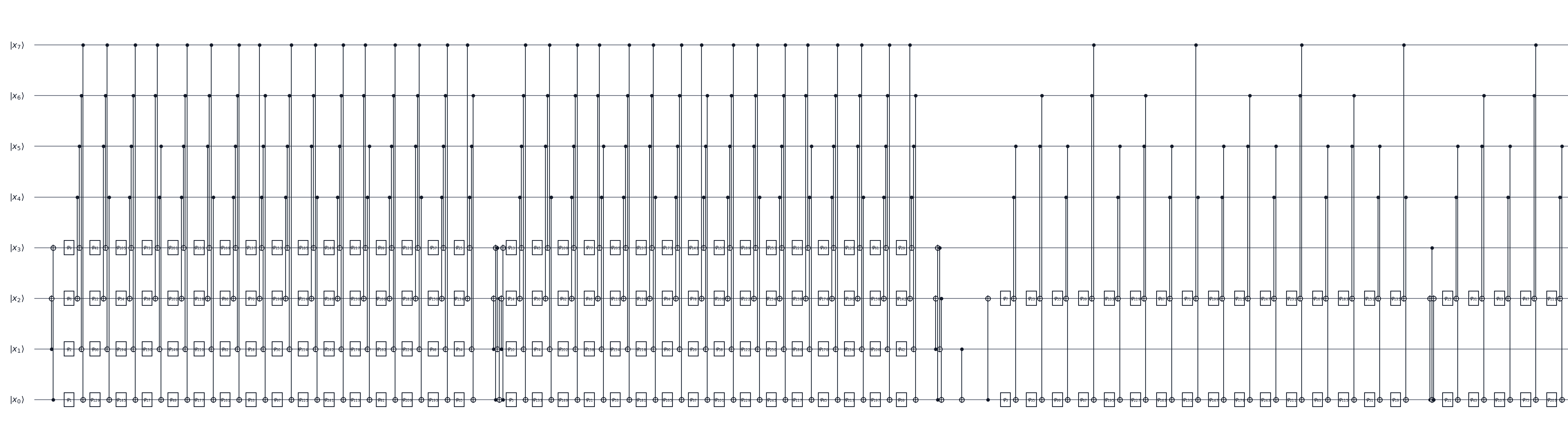}
	
	\vspace{0.6em}
	{\Large $\Longrightarrow$}\quad\small continued
	\vspace{0.6em}
	
	\includegraphics[width=0.9\textheight]{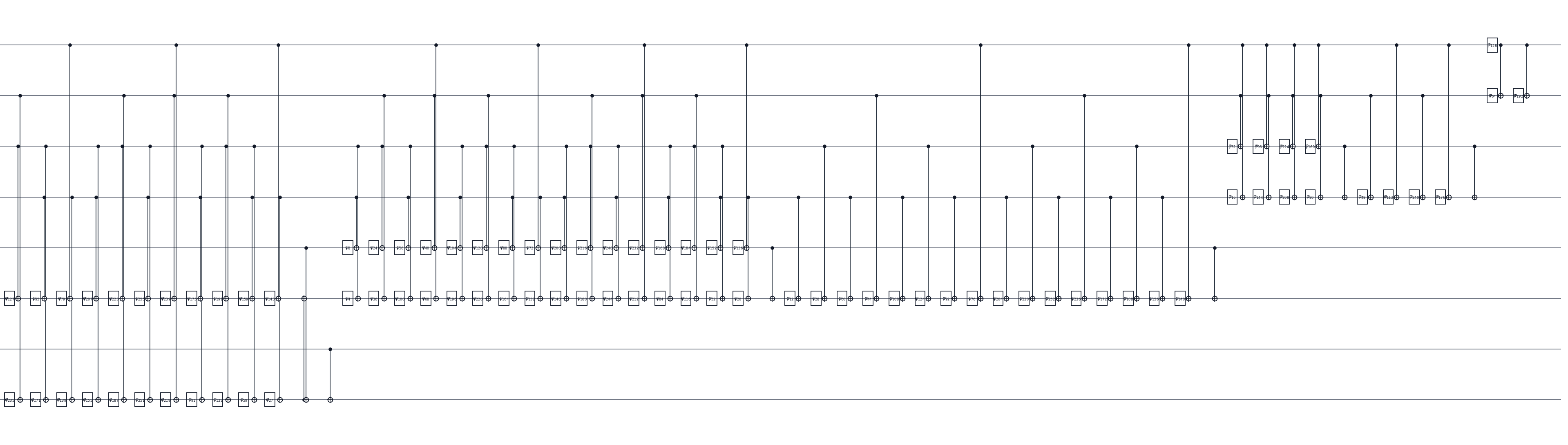}
	
	\caption{Complete expanded GPF$(8)$ circuit, split into two consecutive panels.
		The phase labels use decimal Walsh-mode indices; the lower panel continues from the right end of the upper panel.}
	\label{fig:gpf8-full}
\end{sidewaysfigure}

\begin{sidewaysfigure}[!p]
	\refstepcounter{section}
	\label{app:circuit}

	\addcontentsline{toc}{section}{\protect\numberline{\thesection}Complete Routed GPF Circuit ($n=8$)}

	\centering
	{\Large\bfseries Appendix \thesection\quad Complete Routed GPF Circuit ($n=8$) \par}
	\vspace{1.5em}

	\centering
	\includegraphics[width=0.88\textheight]{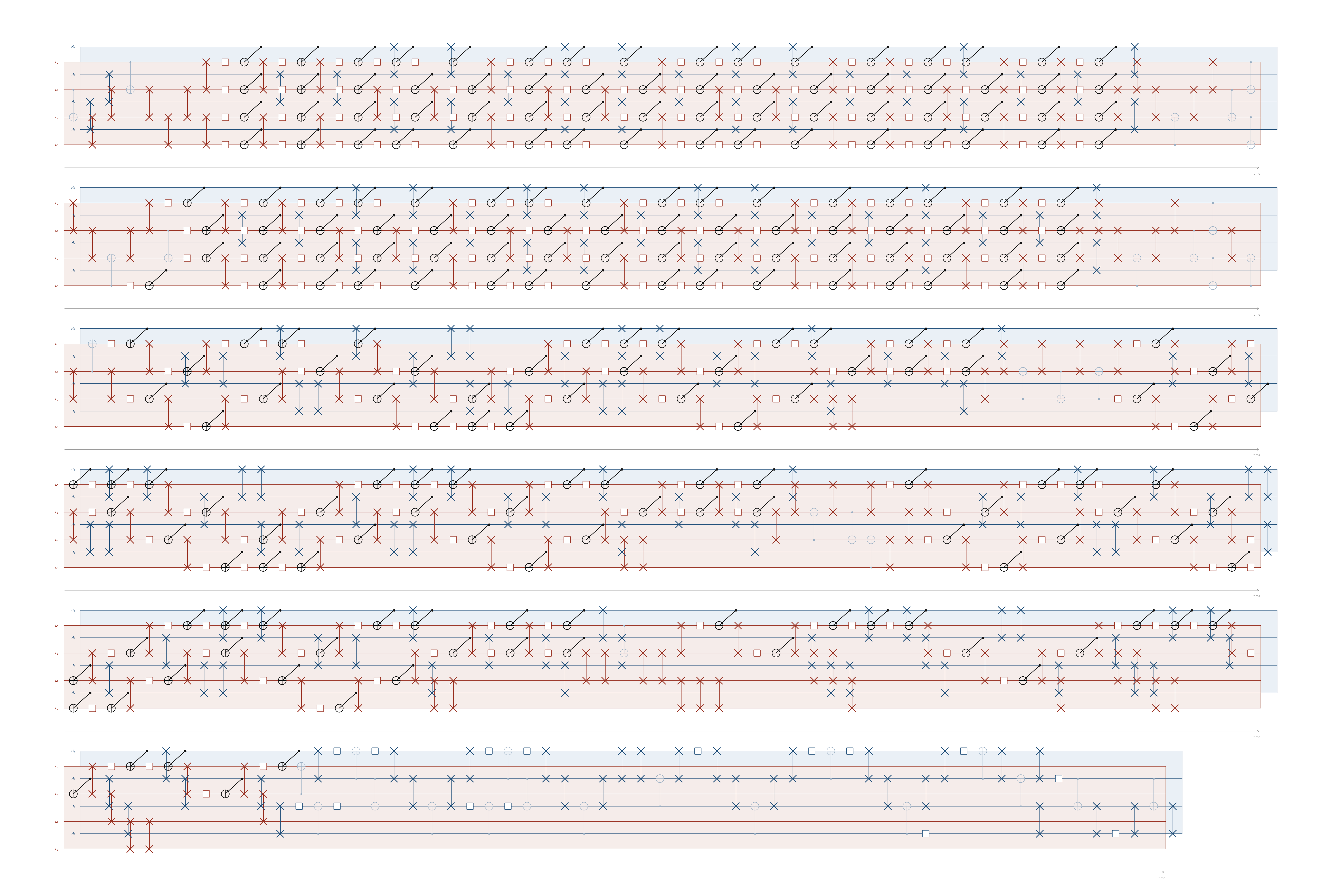}

	\caption{Complete routed $\GPF(8)$ circuit on the $2\times 4$ grid (folded into six segments; low-jump sequence). Dark cross-row links are top-level GPS-2L matchings (Section~\ref{sec:tworow}); light intra-row links are the recursive GPS-1L/GPS$^*$-1L relay (Section~\ref{sec:onerow}). The whole circuit---$278$ CNOTs plus $350$ SWAPs, all nearest-neighbor---realizes the original $\GPF(8)$.}
	\label{fig:gpf-routed-n8}
\end{sidewaysfigure}

\clearpage

\section{Closed-Form Depth Constants}
\label{app:constants}
\setcounter{equation}{0}

We solve the recurrences of Section~\ref{sec:depth} in closed form. Throughout, write
\begin{equation}
    P_m:=\prod_{j=1}^{m}\bigl(1+2^{\,1-2^{j-1}}\bigr),\qquad
    C=\lim_{m\to\infty}P_m\approx 3.40147,
\end{equation}
the limit existing because the $j$-th factor differs from $1$ by $2^{\,1-2^{j-1}}$, which decays doubly exponentially (already $2^{-31}$ at $j=6$, so $P_6$ matches $C$ to nine digits).

\subsection{Balanced subsequence $n=2^\lambda$}
Here $n_H=n_L=2^{\lambda-1}$, both powers of two, so the $\GPF^*$ phase recurrence in~\eqref{eq:recurrences} becomes the scalar recurrence
\begin{equation}
    D^{\GPF^*}_z(2^\lambda)=2^{\,2^{\lambda-1}-1}D^{\GPF^*}_z(2^{\lambda-1})+D^{\GPF^*}_z(2^{\lambda-1})
    =\bigl(2^{\,2^{\lambda-1}-1}+1\bigr)D^{\GPF^*}_z(2^{\lambda-1}).
\end{equation}
With $D^{\GPF^*}_z(1)=1$, telescoping yields
\begin{equation}
    D^{\GPF^*}_z(2^\lambda)=\prod_{j=1}^{\lambda}\bigl(2^{\,2^{j-1}-1}+1\bigr)
    =2^{\,2^\lambda-1-\lambda}\,P_\lambda,
\end{equation}
where the exponent uses $\sum_{j=1}^{\lambda}(2^{j-1}-1)=2^\lambda-1-\lambda$. Therefore
\begin{equation}
    \frac{D^{\GPF^*}_z(2^\lambda)}{2^n/n}=\frac{D^{\GPF^*}_z(2^\lambda)\,2^\lambda}{2^{2^\lambda}}
    =\tfrac12 P_\lambda\;\xrightarrow{\lambda\to\infty}\;\frac{C}{2}\approx 1.70074 .
\end{equation}
The other balanced constants follow from the leading coefficients of the remaining lines of~\eqref{eq:recurrences}. The $\GPF^*$ CNOT recurrence gives $D^{\GPF^*}_{\CNOT}(2^\lambda)=\kappa\,D^{\GPF^*}_z(2^\lambda)+2D^{\GPF^*}_{\CNOT}(2^{\lambda-1})$, whose recursive term is doubly-exponentially smaller than $2^{2^\lambda}$, so $D^{\GPF^*}_{\CNOT}(2^\lambda)/(2^n/n)\to \tfrac{\kappa}{2}C$; for the BRGC fold factor $\kappa=3$ this is $3C/2\approx 5.10221$. The $\GPF$ phase recurrence carries $2^{n_H}$ in place of the $2^{n_H-1}$ of the $\GPF^*$ one, doubling the leading term, so $D^{\GPF}_z(2^\lambda)/(2^n/n)\to 2\cdot(C/2)=C\approx 3.40147$; and the $\GPF$ CNOT recurrence gives $D^{\GPF}_{\CNOT}(2^\lambda)=D^{\GPF}_z(2^\lambda)+\text{(lower order)}$, hence $D^{\GPF}_{\CNOT}(2^\lambda)/(2^n/n)\to C$ as well.

\subsection{Worst case $n=2^\lambda-1$}
Now $n_H=2^{\lambda-1}$ is a perfect power of two while $n_L=2^{\lambda-1}-1$ stays in the same family. Put $g_\lambda:=D^{\GPF^*}_z(2^\lambda-1)$ and $u_\lambda:=g_\lambda/2^{\,2^\lambda-1}$. Using the $\GPF^*$ phase recurrence in~\eqref{eq:recurrences} together with the balanced value $D^{\GPF^*}_z(2^{\lambda-1})=2^{\,2^{\lambda-1}-1-(\lambda-1)}P_{\lambda-1}$ from the previous subsection,
\begin{equation}
    u_\lambda=\tfrac12 u_{\lambda-1}+P_{\lambda-1}\,2^{\,1-\lambda-2^{\lambda-1}} .
\end{equation}
Setting $v_\lambda:=2^\lambda u_\lambda$ removes the factor $\tfrac12$, and since $P_{\lambda-1}\,2^{\,1-2^{\lambda-1}}=P_\lambda-P_{\lambda-1}$ this telescopes:
\begin{equation}
    v_\lambda=v_{\lambda-1}+\bigl(P_\lambda-P_{\lambda-1}\bigr)
    \;\Longrightarrow\;
    v_\lambda=v_2+P_\lambda-P_2 .
\end{equation}
From $D^{\GPF^*}_z(3)=4$ one has $v_2=2$ and $P_2=3$, hence $v_\lambda=P_\lambda-1\to C-1$. Because $D^{\GPF^*}_z(2^\lambda-1)/(2^n/n)=\tfrac{2^\lambda-1}{2^\lambda}\,v_\lambda$,
\begin{equation}
    \frac{D^{\GPF^*}_z(2^\lambda-1)}{2^n/n}\;\xrightarrow{\lambda\to\infty}\;C-1\approx 2.40147 .
\end{equation}
The factor-$3$ and factor-$2$ relations of the previous subsection then give $D^{\GPF^*}_{\CNOT}(2^\lambda-1)/(2^n/n)\to 3(C-1)\approx 7.20441$ and $D^{\GPF}_z(2^\lambda-1)/(2^n/n)=D^{\GPF}_{\CNOT}(2^\lambda-1)/(2^n/n)\to 2(C-1)\approx 4.80294$.

\subsection{Envelopes and the for-all-$n$ bound}
Direct evaluation of~\eqref{eq:recurrences} for all $n\le 8192$ shows that, within each interval $2^\lambda\le n<2^{\lambda+1}$, the normalized depth is minimized at $n=2^\lambda$ and maximized at $n=2^{\lambda+1}-1$, with the two envelopes increasing monotonically toward the limits of the previous subsections. This establishes the for-all-$n$ bound of Corollary~\ref{cor:all-n}.

\subsection{Bounding the Fold Factor $\kappa$}
The $\GPF^*$ CNOT limits carry the fold factor $\kappa$ as $D^{\GPF^*}_{\CNOT}(n)/(2^n/n)\to\tfrac{\kappa}{2}C$ on $n=2^\lambda$ and $\to\kappa(C-1)$ on $n=2^\lambda-1$. The factor is bounded $3\le\kappa\le 4$ for every transition sequence, with $\kappa=3$ for the BRGC and $\kappa\in(3,4)$ for a low-jump sequence; the bound is proved in Appendix~\ref{app:lowjump}. Hence $D^{\GPF^*}_{\CNOT}(n)/(2^n/n)\le 2C$ and $\le 4(C-1)$ respectively: a low-jump sequence raises the $\GPF^*$ CNOT depth by at most a factor $4/3$ over the BRGC value and never changes the asymptotic order.

\clearpage
\section{Mechanism Example Circuits}
\label{app:examples}

This appendix collects illustrative single-block figures referenced from Sections~\ref{sec:placement}--\ref{sec:routing}: the logical circuits to be routed (Figs.~\ref{fig:logical-n8} and \ref{fig:star-logical-n8}), the 2.5D rendering of the routed top-level GPS (Fig.~\ref{fig:routed-3d-n8}), and the two further single-row routing examples (Figs.~\ref{fig:star-1row}--\ref{fig:segment-1256}).

\begin{figure}[htbp]
    \centering
    \includegraphics[width=0.88\textwidth]{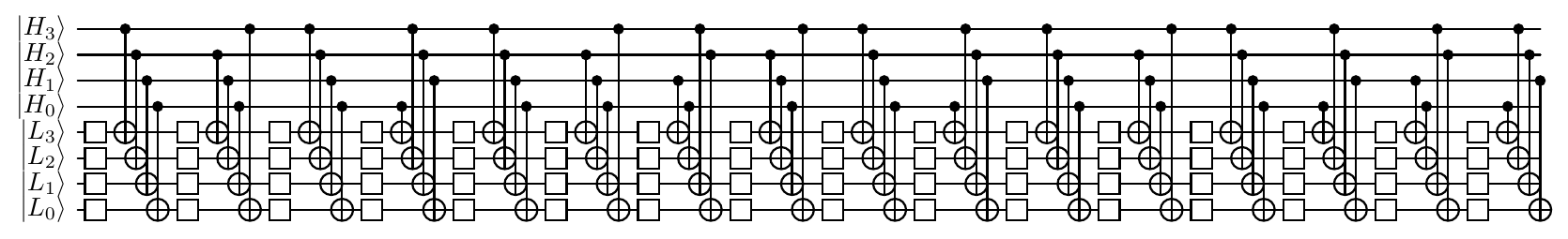}
    \caption{Logical top-level GPS to be routed ($n=8$): the cross-block CNOT spans vary with the transition sequence, so most are non-nearest-neighbor---this is why routing is needed.}
    \label{fig:logical-n8}
\end{figure}

\begin{figure}[htbp]
    \centering
    \includegraphics[width=0.88\textwidth]{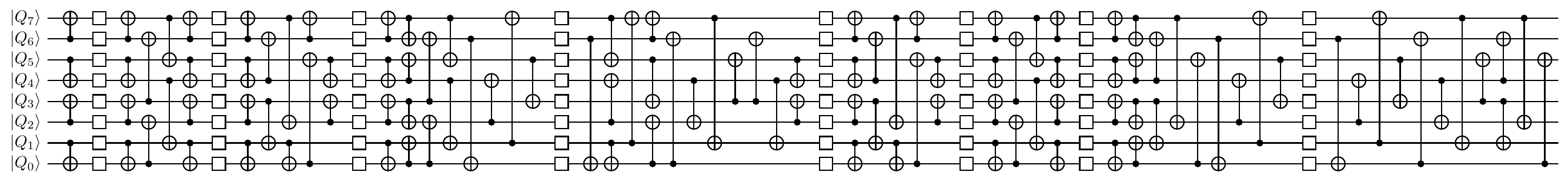}
    \caption{Logical single-row GPS$^*$ to be routed ($n=8$): the CNOT spans are highly irregular, and the long-range pairs are exactly what single-row routing must decompose to nearest-neighbor.}
    \label{fig:star-logical-n8}
\end{figure}

\begin{figure}[htbp]
    \centering
    \includegraphics[width=0.95\columnwidth]{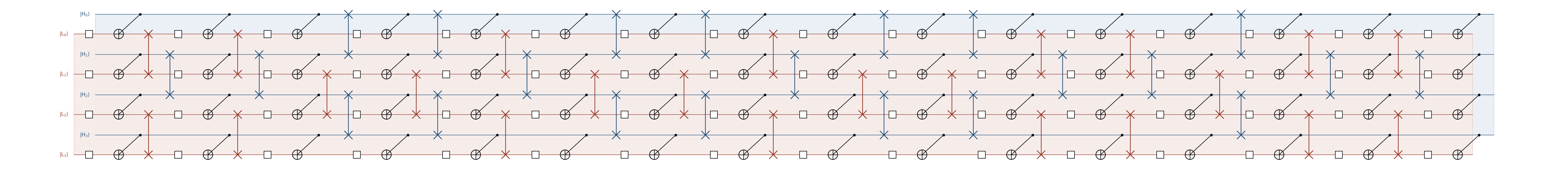}
    \caption{2.5D depth chart of the routed top-level GPS block ($n=2\times 4$, same low-jump schedule as Fig.~\ref{fig:routed-n8}).}
    \label{fig:routed-3d-n8}
\end{figure}

\begin{figure}[htbp]
    \centering
    \includegraphics[width=0.95\columnwidth]{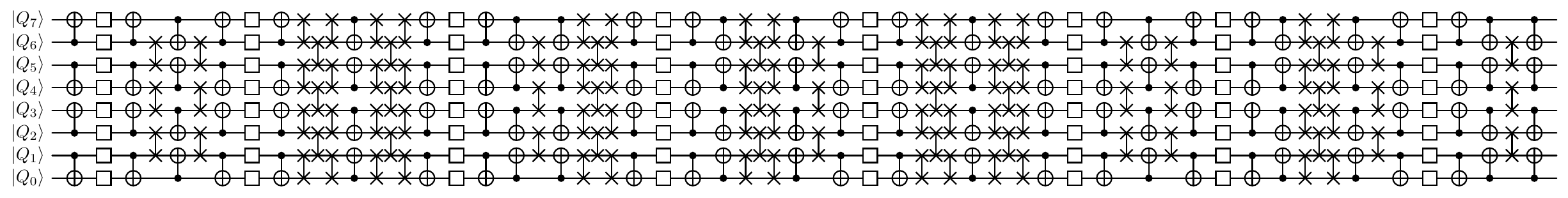}
    \caption{Example 2: routed single-row GPS$^*$ module ($n=8$); same folding and brickwork as Fig.~\ref{fig:gps-1row}, with the schedule switched to GPS$^*$.}
    \label{fig:star-1row}
\end{figure}

\begin{figure}[htbp]
    \centering
    \includegraphics[width=0.8\columnwidth]{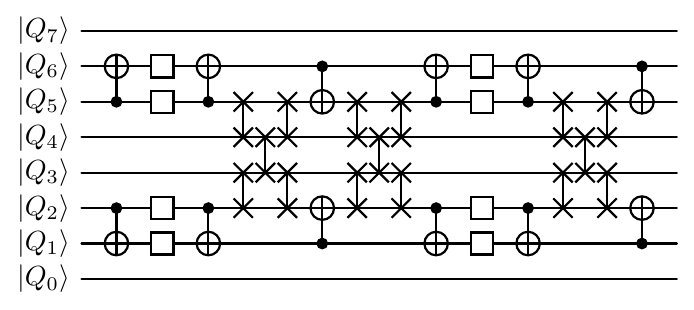}
    \caption{Example 3: single-row GPS$^*$ on segment $\{1,2,5,6\}$ ($n=8$): the remote SWAP$(2,5)$ is brickwork-decomposed through bystanders $\{3,4\}$ and restored---a deeper-recursion intra-group SWAP becoming remote. This realizes GPS$^*$(4).}
    \label{fig:segment-1256}
\end{figure}

\section{Closed-Form Derivation of Complete GPF Routing Complexity}
\label{app:route-constants}
\setcounter{equation}{0}

This appendix fills in the algebraic details of Theorems~\ref{thm:route-depth} and \ref{thm:route-count}, reusing $P_m=\prod_{j=1}^{m}(1+2^{\,1-2^{j-1}})$ and $C=\lim_{m\to\infty}P_m\approx 3.40147$ from Appendix~\ref{app:constants}.

\subsection{Leading-Term Separation}

From the recurrence of Theorem~\ref{thm:route-depth}, $D^{\GPF\text{-}2\mathrm{L}}_{\mathrm{route}}(n)$ consists of the top-level term $3\mu_{\mathrm{jump}}2^{n_H}D^{\GPF^*}_z(n_L)$ and two single-row recursions $D^{\GPF\text{-}1\mathrm{L}}_{\mathrm{route}}(n_H,0)$, $D^{\GPF^*\text{-}1\mathrm{L}}_{\mathrm{route}}(n_L,0)$. The latter act on approximately $n/2$ qubits, with magnitude $\Theta(2^{n/2})$, doubly-exponentially smaller than the top-level's $\Theta(2^n/n)$ (relative correction $\sim 2^{-n/2}$); under normalization they vanish, so the leading constant is determined entirely by the top-level term.

\subsection{Balanced Subsequence $n=2^\lambda$}

Here $n_H=n_L=2^{\lambda-1}$, both powers of two. By the balanced-subsequence telescoping of Appendix~\ref{app:constants},
\begin{equation}
    D^{\GPF^*}_z(2^\lambda)=\prod_{j=1}^{\lambda}\bigl(2^{\,2^{j-1}-1}+1\bigr)
    =2^{\,2^\lambda-1-\lambda}\,P_\lambda,
    \qquad\text{hence}\quad
    D^{\GPF^*}_z(2^{\lambda-1})=2^{\,2^{\lambda-1}-\lambda}\,P_{\lambda-1}.
\end{equation}

Substituting into the top-level term:
\begin{equation}
    3\mu_{\mathrm{jump}}\,2^{n_H}D^{\GPF^*}_z(n_L)
    =3\mu_{\mathrm{jump}}\,2^{\,2^{\lambda-1}}\cdot 2^{\,2^{\lambda-1}-\lambda}\,P_{\lambda-1}
    =3\mu_{\mathrm{jump}}\,2^{\,2^\lambda-\lambda}\,P_{\lambda-1}.
\end{equation}

Dividing by $2^n/n=2^{\,2^\lambda-\lambda}$ yields $D^{\GPF\text{-}2\mathrm{L}}_{\mathrm{route}}(2^\lambda)/(2^n/n)\to 3\mu_{\mathrm{jump}}P_{\lambda-1}\to 3\mu_{\mathrm{jump}}C\approx 10.82$.

\subsection{Upper Envelope $n=2^\lambda-1$}

Here $n_H=2^{\lambda-1}$ is a power of two and $n_L=2^{\lambda-1}-1$ remains in the same family, reusing the worst-case telescoping of Appendix~\ref{app:constants}. Set $g_\lambda=D^{\GPF^*}_z(2^\lambda-1)$, $u_\lambda=g_\lambda/2^{2^\lambda-1}$, and $v_\lambda:=2^\lambda u_\lambda$. The phase recurrence at $n=2^\lambda-1$ gives $u_\lambda=u_{\lambda-1}+2^{-(2^{\lambda-1})}\,$-type increments; using the identity $P_{\lambda-1}\,2^{1-2^{\lambda-1}}=P_\lambda-P_{\lambda-1}$ (immediate from the definition of $P_m$), these telescope to $v_\lambda=v_{\lambda-1}+(P_\lambda-P_{\lambda-1})$, whence $v_\lambda=P_\lambda-1\to C-1$. Accounting for the GPF phase recurrence's $2^{n_H}$ factor (in place of $2^{n_H-1}$, contributing a factor of $2$), the upper envelope is $D^{\GPF\text{-}2\mathrm{L}}_{\mathrm{route}}(2^\lambda-1)/(2^n/n)\to 3\mu_{\mathrm{jump}}\cdot 2(C-1)\approx 15.27$. For all $n$, the normalized routing depth is bracketed between $3\mu_{\mathrm{jump}}C$ and $3\mu_{\mathrm{jump}}\cdot 2(C-1)$.

\subsection{Gate Count}

The gate-count recurrence of Theorem~\ref{thm:route-count} is isomorphic to the depth recurrence, with each depth layer carrying an additional width factor of approximately $n_H-1$ parallel SWAPs. Under normalization by $2^n$ (rather than $2^n/n$), this width factor is absorbed into the denominator, halving the constant to $\tfrac{3}{2}\mu_{\mathrm{jump}}$; the lower and upper envelopes are thus $\tfrac{3}{2}\mu_{\mathrm{jump}}C\approx 5.41$ and $\tfrac{3}{2}\mu_{\mathrm{jump}}\cdot 2(C-1)\approx 7.64$.

Direct recurrence evaluation for all $n\le 8192$ confirms: within each interval $2^\lambda\le n<2^{\lambda+1}$, normalized depth and gate count are minimized at $n=2^\lambda$ and maximized at $n=2^{\lambda+1}-1$, with both envelopes monotonically approaching the limits above.

\section{The Low-Jump Transition Sequence: Construction and Constants}
\label{app:lowjump}
\setcounter{equation}{0}

Viewed as a Gray code on the $n_H$ high coordinates, the transition sequence $(f_0,\ldots,f_{2^{n_H}-1})$ of Section~\ref{sec:jump-cost} is a \emph{graph-compatible} code: its routing cost $J$ (Theorem~\ref{thm:gps2l}) is the total movement of the flipped coordinate along the position ring, so a sequence whose consecutive flips fall in adjacent ring positions is a cyclic, ring analogue of the skew-tolerant (path-compatible) Gray codes studied by Wilson--Blaum and Sac~Himelfarb--Schwartz \cite{SacHimelfarb2025SkewTolerant}. We minimize the jump density $\mu_{\mathrm{jump}}$ of Eq.~\eqref{eq:mu-def}. This appendix gives the construction behind Lemma~\ref{lem:minjump} and the curves of Fig.~\ref{fig:gpfstar-depth-scaling}, and bounds the synthesis-side fold factor $\kappa$.

\subsection{Validity and Lower Bound}
A candidate is admissible only as a genuine closed Gray cycle: applying all $2^{n_H}$ flips must traverse every high-space mode once and return to mode $0$, the same criterion enforced on the synthesis side. Length $2^{n_H}$ or a final entry $0$ does not suffice---appending a closing $0$ to an open Gray path fails for most $n_H$ ($2,3,4,6,8,\ldots$), since the result does not return to mode $0$. Because adjacent flips differ, each ring step has $\Dist_{n_H}\ge 1$, giving the hard bound $J\ge 2^{n_H}$, i.e.\ $\mu_{\mathrm{jump}}\ge 1$. Equality---a perfectly adjacent, ``$1$-skew-tolerant'' cycle---is attainable only for small $n_H$: complete path-compatible Gray codes are absent at $n_H=7$ and conjectured absent for all $n_H\ge 7$ \cite{SacHimelfarb2025SkewTolerant}, so $\mu_{\mathrm{jump}}$ necessarily exceeds $1$ beyond a few coordinates.

\subsection{Construction}
For $n_H\le 6$ we obtain the optimal cycle by exact search: iterative deepening on the target cost upward from $2^{n_H}$, returning the first feasible target, every smaller one having been exhaustively certified infeasible (an admissible remaining-cost bound---each remaining step costs at least $1$---and fixing the first flip by rotational symmetry keep this tractable). The optima are $\mu_{\mathrm{jump}}=1.0$ for $n_H\le 5$ (e.g.\ $n_H{=}2$: $0,1,0,1$;\ \ $n_H{=}3$: $0,1,0,2,0,1,0,2$) and $\mu_{\mathrm{jump}}=66/64\approx 1.03$ for $n_H=6$. These are cached.

Beyond $n_H=6$ exact certification is impractical, and the optimal cycles lose the clean self-similar form they keep through $n_H=5$ (the $n_H=6$ optimum already needs one distance-$3$ step). We extend deterministically by a \emph{reflected sandwich}, the reflect-and-shift paradigm of supercomposite Gray codes \cite{SacHimelfarb2025SkewTolerant} specialized to the ring. One new coordinate $c$, placed at a ring position adjacent to the existing block, splits the order-$r$ cycle $B_r$ into two equal halves stacked along $c$:
\begin{equation*}
B_r=
\left[\begin{array}{c|c}
\;B_{r-1}\; & 0\\[2pt]\hline\rule{0pt}{2.4ex}
\;\widetilde{B}_{r-1}\; & 1
\end{array}\right]
\qquad
\begin{array}{@{}l}
\text{block 1: forward sweep of the }(r{-}1)\text{-subcube},\ c=0,\\[4pt]
\text{block 2: reflected sweep }\widetilde{B}_{r-1},\ c=1,
\end{array}
\end{equation*}
where the right column is the value of the new coordinate $c$. Coordinate $c$ stays $0$ through a forward sweep $B_{r-1}$ of the $(r{-}1)$-subcube, flips once to enter the reflected sweep $\widetilde{B}_{r-1}$ at $c=1$, and flips once more to close---so $c$ changes exactly twice, and because $c$ is ring-adjacent to the block, both junctions are unit steps. The smallest instance is $r=3$ (new coordinate $2$), with transition sequence $0,1,0,\mathbf{2},\,0,1,0,\mathbf{2}$, the two $\mathbf{2}$'s marking the block boundaries; the $n_H=7$ cycle has exactly this shape, two sweeps of the $6$-subcube joined by the two flips of coordinate~$6$.

The base need not be a single level below: the split that minimizes cost is highly regular---for $7\le n_H\le 10$ the recursion rests directly on the exhaustively-optimal $k=6$ block, and for $n_H\ge 11$ it adds one coordinate for odd $n_H$ and three for even, the two choices differing by only an $O(1)$ boundary term. Up to that term the extension is single-coordinate and uniquely determined. It closes for every $n_H$, with $\mu_{\mathrm{jump}}$ falling from $1.0625$ at $n_H=7$ to $\approx 1.060$ by $n_H=16$; this is the $\approx 1.06$ of Lemma~\ref{lem:minjump}. It is not proven optimal, so ``low-jump'' is descriptive.

\subsection{The Two Constants: Jump Density and Fold Factor}
\label{app:kappa-bound}
The sequence enters the two halves of the paper through two opposed functionals (Section~\ref{sec:graycode}). On the compilation side the \emph{jump density} $\mu_{\mathrm{jump}}$ (Eq.~\eqref{eq:mu-def}) fixes the routing constant $3\mu_{\mathrm{jump}}C$ (Theorem~\ref{thm:gps2l}); the low-jump value $\approx 1.06$ gives $3\mu_{\mathrm{jump}}C\approx 10.82$, whereas the BRGC's $\mu_{\mathrm{jump}}$ rises toward $2$ as $n_H$ grows.

On the synthesis side the same sequence sets the \emph{fold factor} $\kappa$ of Eq.~\eqref{eq:kappa-def}, the $\GPS^*$ CNOT-to-phase layer ratio. The $\GPS^*$ schedule processes the $2^{n_H}$ flips in consecutive groups of four; each group retains two phase layers and emits either six or eight CNOT layers---the two branches of the forward--reverse fold, according to whether the group is matched. The per-group ratio is thus $3$ or $4$, and summing over groups gives
\begin{equation}
    3\le\kappa\le 4,
\end{equation}
with $\kappa=3$ exactly when every group is matched. The BRGC, whose odd-step flips are all the lowest coordinate, matches every group and attains $\kappa=3$; a low-jump sequence mixes the two cases, so its computed per-width fold factor $\kappa_{\mathrm{lj}}(r_c)$ lies in $(3,4)$ and converges to $\kappa^\ast\approx 3.64$ (Fig.~\ref{fig:gpfstar-depth-scaling}). Consequently $D^{\GPF^*}_{\CNOT}(n)/(2^n/n)\to\tfrac{\kappa}{2}C\le 2C$ on $n=2^\lambda$ and $\to\kappa(C-1)\le 4(C-1)$ on $n=2^\lambda-1$: a low-jump sequence raises the $\GPF^*$ CNOT depth by at most a factor $4/3$ over the BRGC value and never changes the asymptotic order.

\section{Proof of the Offset-Ring Lemma}
\label{app:offset-ring}
\setcounter{equation}{0}

We prove Lemma~\ref{lem:offset-ring}. Identify each column $j\in\{0,\dots,n_H-1\}$ with the \emph{folded-cycle coordinate} $m=\sigma(j)\in\mathbb{Z}_{n_H}$, where $\sigma(2k)=k$ and $\sigma(2k+1)=n_H-1-k$; this is a bijection, and in the canonical arrangement the label carried at folded coordinate $m$ on each row is exactly $m$.

A swap layer permutes the contents of a row by its brick involution. In folded coordinates the even-brick map $E$ (pairs $(0,1),(2,3),\dots$) and the odd-brick map $O$ (pairs $(1,2),(3,4),\dots$) act as the two reflections of $\mathbb{Z}_{n_H}$,
\begin{equation}
    \hat E(m)=-1-m,\qquad \hat O(m)=-m \pmod{n_H},
\end{equation}
which is the pairwise content of $\sigma(E(j))=n_H-1-\sigma(j)$ and $\sigma(O(j))=-\sigma(j)$ (checked on each brick pair, including the fixed boundary column when $n_H$ is odd).

Track each row as a map (folded coordinate $\mapsto$ label) of the affine form $m\mapsto\varepsilon m+c$ with $\varepsilon\in\{+1,-1\}$ and $c\in\mathbb{Z}_{n_H}$; both rows start at $(\varepsilon,c)=(+1,0)$. Applying a brick whose reflection is $\hat B(m)=\varepsilon_B m+c_B$---so $(\varepsilon_E,c_E)=(-1,-1)$ and $(\varepsilon_O,c_O)=(-1,0)$---replaces $m\mapsto\varepsilon m+c$ by $m\mapsto\varepsilon\,\hat B(m)+c=(\varepsilon\varepsilon_B)m+(\varepsilon c_B+c)$. Both bricks are reflections ($\varepsilon_B=-1$), so after $t$ layers both rows share the same parity $\varepsilon^H_t=\varepsilon^L_t=(-1)^t$. Hence the offset
\begin{equation}
    \delta_t(m)=\bigl(\varepsilon^L_t m+c^L_t\bigr)-\bigl(\varepsilon^H_t m+c^H_t\bigr)=c^L_t-c^H_t
\end{equation}
is independent of $m$: every layout reached by this schedule is a uniform offset state.

For the offset value, SWAP-A applies $E$ to the upper row and $O$ to the lower row, and SWAP-B applies $O$ to the upper and $E$ to the lower. Writing $c_{t+1}=\varepsilon_t c_B+c_t$ with $\varepsilon_t=(-1)^t$, both the SWAP-A step ($t$ even) and the SWAP-B step ($t$ odd) give
\begin{equation}
    \delta_{t+1}=c^L_{t+1}-c^H_{t+1}=\delta_t+1 \pmod{n_H}.
\end{equation}
Thus alternating SWAP-A/SWAP-B from $s_0$ yields $\delta_t=t$, i.e.\ $s_0\to s_1\to\cdots\to s_{n_H-1}\to s_0$ with each offset visited once. The same update shows that from any uniform state a single SWAP-A or SWAP-B layer reaches a uniform state of offset $\delta\pm1$, the two layer types giving opposite signs; the offset states therefore form a $\mathbb{Z}_{n_H}$-cycle, and the least number of layers from $s_a$ to $s_b$ is the cycle distance $\Dist_{n_H}(a,b)=\min(|a-b|,n_H-|a-b|)$. This proves Lemma~\ref{lem:offset-ring}.

\end{document}